\documentclass[11pt, a4paper]{amsart}
\usepackage{amssymb, amsmath, amscd, bm, amsthm, pdfpages, setspace, hyperref, mathtools, subcaption, graphicx, enumerate, cite}

\theoremstyle{plain}
\newtheorem{theorem}{Theorem}[section]
\newtheorem{proposition}[theorem]{Proposition}
\newtheorem{definition}[theorem]{Definition}

\newtheorem{assumption}[theorem]{Assumption}
\newtheorem{example}[theorem]{Example}
\newtheorem{lemma}[theorem]{Lemma}
\newtheorem{remark}[theorem]{Remark}

\newtheorem{algorithm}[theorem]{Algorithm}

\newcommand{\diag}{\mathrm{diag}}
\newcommand{\ave}{\mathrm{ave}}
\newcommand{\vertiii}[1]{{\vert\kern-0.25ex\vert\kern-0.25ex\vert #1 
		\vert\kern-0.25ex\vert\kern-0.25ex\vert}}
\newcommand{\bvertiii}[1]{{\big\vert\kern-0.25ex\big\vert\kern-0.25ex\big\vert #1 
		\big\vert\kern-0.25ex\big\vert\kern-0.25ex\big\vert}}

\def\@Rref#1{\hbox{\rm \ref{#1}}}
\def\Rref#1{\@Rref{#1}}

\theoremstyle{plain}

\begin{document}

\title[Self-triggered Consensus with Quantized Measurements]{Self-triggered Consensus of
	Multi-agent Systems with Quantized Relative State Measurements}

\thispagestyle{plain}

\author{Masashi Wakaiki}
\address{Graduate School of System Informatics, Kobe University, Nada, Kobe, Hyogo 657-8501, Japan}
 \email{wakaiki@ruby.kobe-u.ac.jp}
 \thanks{This work was supported by JSPS KAKENHI Grant Number JP20K14362.}

\begin{abstract}
This paper addresses the consensus problem of
first-order continuous-time multi-agent systems over 
undirected graphs.
Each agent samples relative state measurements
in a self-triggered fashion and transmits
the sum of the measurements to its neighbors.
Moreover, we use finite-level dynamic quantizers and apply
the zooming-in technique.
The proposed joint design method for quantization and self-triggered sampling
achieves asymptotic consensus, and inter-event times are strictly positive.
Sampling times are determined explicitly with iterative procedures
including the computation of the Lambert $W$-function.
A simulation example is provided to illustrate the effectiveness of
the proposed method.
\end{abstract}

\maketitle

\section{Introduction}
With the recent development of 
information and communication technologies,
multi-agent systems have received 
considerable attention.
Cooperative control of multi-agent systems can be applied to various
areas such as multi-vehicle formulation \cite{Fax2004} and 
distributed sensor networks \cite{Olfati-Saber2005}.
A basic coordination problem of multi-agent systems is consensus,
whose aim is to reach an agreement on the
states of all agents.
A theoretical framework
for consensus problems has been
introduced in the seminal work~\cite{Olfati-Saber2004}, and 
substantial progress has been made since then; see 
the survey papers~\cite{Olfati-Saber2007,Chen2019Survey} and the references therein.

In practice, digital devices are used in
multi-agent systems.
Conventional approaches to implementing digital platforms involve periodic sampling.
However, periodic sampling can lead to
unnecessary control updates and state measurements, 
which are undesirable for resource-constrained multi-agent systems.
Event-triggered control~\cite{Arzen1999, Tabuada2007, Heemels2008}
and self-triggered control~\cite{Velasco2003, Wang2009, Anta2010} are
promising alternatives to traditional periodic control.
In both event-triggered and self-triggered control systems,
data transmissions and control updates occur only when needed.
Event-triggering mechanisms use 
current measurements and check triggering conditions
continuously or periodically.
On the other hand, self-triggering mechanisms avoid
such frequent monitoring by calculating
the next sampling time when data are obtained.
Various methods have been developed  
for event-triggered consensus and self-triggered consensus; see, e.g., \cite{Dimarogonas2012, Seyboth2013, Fan2015, Yi2019, Liu2022SC}.
Comprehensive surveys on this topic are available in~\cite{Ding2018, Nowzari2019}.
Some specifically relevant studies are cited below.

The bandwidth of communication channels
and the accuracy of sensors may be limited in multi-agent systems.
In such situations, only imperfect information 
is available to the agents.
We also face the theoretical question of 
how much accuracy in information is necessary for consensus.
From both practical and theoretical point of view,
quantized consensus has been studied extensively.
For continuous-time multi-agent systems,
infinite-level static quantization is often considered
under the situation where quantized measurements are 
obtained continuously; see, e.g., \cite{Dimarogonas2010, Ceragioli2011, Liu2013, Guo2013, Wu2014IJRNC, Li2018, Xu2022}.
Event-triggering mechanisms and self-triggering mechanisms
have been proposed for continuous-time 
multi-agent systems with infinite-level static quantizers in 
\cite{Garcia2013IJC, Zhang2015, Zhang2016TCSII, Yi2016, Liu2016MAS,
	Wu2018TCSI, Dai2019, Li2021, Dai2022, Wang2023}.
Self-triggered consensus with ternary controllers
has been also studied in \cite{Persis2013, Matsume2021}.
For 
event-triggered consensus under unknown input delays,
finite-level dynamic quantizers have been
developed in \cite{Golestani2022}, where the quantization error goes to zero
as the agent state converges to the origin.

For discrete-time multi-agent systems, finite-level
dynamic quantizers  to achieve asymptotic consensus
have been designed in~\cite{Carli2010, Li2011, You2011, Li2014, Qui2016}.
This type of dynamic quantization has been also used for
periodic sampled-data consensus~\cite{Ma2018},
event-triggered consensus~\cite{Chen2017, Ma2018TSMCS, Yu2019, Lin2022,
	Lin2022ISA}, and consensus under denial-of-service attacks~\cite{Feng2022}.
Moreover, an event-triggered average consensus protocol 
has been proposed for 
discrete-time multi-agent systems with integer-valued states 
in \cite{Rikos2021}, and it has been extended to the
privacy-preserving case in \cite{Rikos2023}.

In this paper, we consider first-order continuous-time  multi-agent systems
over undirected graphs. Our goal is to jointly design
a finite-level dynamic quantizer
and  a self-triggering mechanism for asymptotic consensus.
We focus on the situation where \textit{relative} states, not \textit{absolute} states,
are sampled as, e.g., in~\cite{Liu2022SC, Dimarogonas2010, Guo2013, Liu2016MAS, Yi2016, Li2018, Wu2018TCSI, Dai2019,
	Dai2022}.
We assume that each agent's sensor
has a scaling parameter to adjust the maximum measurement
range and the accuracy.
For example, if indirect time-of-flight sensors~\cite{Horaud2016}  are installed in agents,
then the modulation frequency of light signals determines
the maximum range and the accuracy.
In the case of cameras, they can be changed by adjusting
the focal length; see Section~11.2 of \cite{Spong2020}
for a mathematical model of cameras.

In the proposed self-triggered framework,
the agents send the sum of the relative state measurements
to all their neighbors as
in the self-triggered consensus algorithm presented
in~\cite{Fan2015}.
In other words, each agent communicates with its neighbors only at 
the sampling times of itself and its neighbors.
The sum is transmitted so that the neighbors compute the next sampling times, not the inputs. After receiving it, 
the neighbors update the next sampling times.
Since the measurements are already quantized when they are sampled,
the sum can be transmitted without error, 
even over channels with finite capacity.

The main contributions of this paper are summarized as follows:

1. We 
propose a joint algorithm for finite-level dynamic quantization
and self-triggered  sampling of the relative states.
We also provide a sufficient condition for the
consensus of the quantized
self-triggered multi-agent system.
This sufficient condition represents a quantitative trade-off
between data accuracy and sampling frequency.
Such a trade-off can be a useful guideline for
sensing performance, power consumption, and channel capacity.

2. In the proposed method, the inter-event times, i.e., the 
sampling intervals of each agent, are strictly positive, and hence Zeno behavior does not occur.
In addition, the agents can compute sampling times using an
explicit formula with the Lambert $W$-function (see, e.g.,~\cite{Corless1996}
for the Lambert $W$-function). Consequently,
the proposed self-triggering mechanism is simple and efficient
in computation.

We now compare our results with previous studies. The
finite-level dynamic quantizers developed in \cite{Carli2010, Li2011, You2011, Li2014, Qui2016} and their aforementioned extensions
require the
\textit{absolute} states. More specifically, they quantize the error
between the absolute state and its estimate for 
communication over finite-capacity channels.
In this framework, 
the agents have to estimate the states of all their
neighbors for decoding.
In contrast, we develop finite-level dynamic quantizers for 
\textit{relative} state measurements.
As in the existing studies above,
we also employ the zooming-in technique
introduced for single-loop systems in \cite{Brockett2000, Liberzon2003Automatica}.
However, due to the above-mentioned difference in what is quantized,
the quantizer we study has several notable features. For example,
the proposed algorithm can be applied to
GPS-denied environments. Moreover,
the estimation of neighbor states is not needed,
which reduces the computational burden on the agents.

A finite-level quantizer may be saturated, i.e., it  does not 
guarantee the accuracy of quantized data in general if
the original data is outside of the quantization region.
To achieve asymptotic consensus,
we need to update the scaling parameter
of the quantizer so that
the relative state measurement is within the quantization region
and the quantization error  goes to zero  asymptotically.
In~\cite{Yi2016, Dai2019, Dai2022},
infinite-level static quantizers have been used  for
quantized self-triggered consensus of first-order multi-agent systems.
Hence the issue of quantizer saturation has not been addressed there.
In \cite{Yi2016},
infinite-level uniform quantization has been considered,
and consequently only consensus to a bounded region
around the average of the agent states has been
achieved.
The quantized self-triggered control algorithm 
proposed in \cite{Dai2019, Dai2022}
achieves asymptotic consensus with the help of 
infinite-level logarithmic quantizers, but
sampling times have to belong to the set $\{
t = kh: \text{$k$ is a nonnegative integer}
\}$ with some $h >0$, which makes it easy to exclude
Zeno behavior. 
Table~\ref{table:comparison} summarizes the comparison between
this study and several relevant studies.

The difficulty of  this study is that the following three 
conditions
must be satisfied:
\begin{itemize}
	\item avoiding quantizer saturation;
	\item decreasing the quantization error asymptotically; and
	\item guaranteeing that the inter-event times are strictly positive.
\end{itemize}
To address this difficulty, we  introduce a new
semi-norm $\vertiii{\cdot}_{\infty}$ for the analysis of multi-agent systems.
The semi-norm is constructed from the maximum norm and is suitable for
handling errors of individual agents due to quantization 
and self-triggered sampling.
Moreover, the Laplacian matrix 
$L\in \mathbb{R}^{N \times N}$ 
of the multi-agent system has
the following semi-contractivity property: There exists a constant $\gamma >0$ such that
\[
\vertiii{e^{-Lt} v}_{\infty} \leq e^{-\gamma t} \vertiii{v}_{\infty}
\]
for all $v \in \mathbb{R}^N$ and $t \geq 0$; 
see \cite{Jafarpour2022,Pasquale2022} for the semi-contraction theory.
The semi-contractivity property
facilitates the analysis of state trajectories under self-triggered sampling
and consequently leads to a simple design of the scaling parameter for finite-level
dynamic quantization.

\begin{table*}[b]
	\renewcommand{\arraystretch}{1.2}
	\caption{Comparison between 
		this study and several relevant studies.}
	\label{table:comparison}
	\centering
	\begin{tabular}{r|llll} 
		& Triggering mechanism & Measurement & Quantization & Agent dynamics \\ \hline 
		This study & Self-trigger & Relative state   & Finite-level \& dynamic &  First-order \\
		\cite{Golestani2022, Ma2018TSMCS, Yu2019, Lin2022,Lin2022ISA} & Event-trigger    & Absolute state  & Finite-level \& dynamic  & 	
		High-order \\
		\cite{Yi2016, Dai2019, Dai2022} & Self-trigger & Relative state & Infinite-level \& static & First-order 
	\end{tabular}
\end{table*}

The rest of this paper is organized as follows.
In Section~\ref{sec:sys_model}, we introduce the system model.
In Section~\ref{sec:preliminaries}, we provide
some preliminaries on the semi-norm and sampling times.
Section~\ref{sec:Analysis} contains the main result, which
gives a sufficient condition for consensus.
In Section~\ref{sec:self_triggered_computation}, we explain
how the agents compute sampling times in a self-triggered fashion.
A simulation example is given  in
Section~\ref{sec:numerical_sim}, and
Section~\ref{sec:conclusion} concludes this paper.

Notation: 
We denote
the set of nonnegative integers by $\mathbb{N}_0$.
We define $\inf \emptyset \coloneqq \infty$.
Let $M,N \in \mathbb{N}$. We denote the transpose of $A \in
\mathbb{R}^{M \times N}$ by $A^{\top}$.
For a vector $v \in \mathbb{R}^N$ with the $i$-th element $v_i$, 
its maximum norm is 
\[\|v\|_{\infty} \coloneqq 
\max \{
|v_1|,\dots,|v_N|
\},
\] and the corresponding induced norm of $A \in \mathbb{R}^{M \times N}$
with the \mbox{$(i,j)$-th} element $A_{ij}$ is given by
\[
\|A\|_{\infty} = \max \left\{ 
\sum_{j=1}^N |A_{ij}| : 1 \leq i \leq M
\right\}.
\]
When 
the eigenvalues $\lambda_1,\lambda_2,\cdots,\lambda_N \in \mathbb{R}$ of 
a symmetric matrix $P\in \mathbb{R}^{N\times N}$ with $N \geq 2$
satisfy $\lambda_1 \leq \lambda_2 \leq 
\cdots \leq \lambda_N$, we write $\lambda_2(P) \coloneqq \lambda_2$.
We define 
\[
\mathbf{1} \coloneqq 
\begin{bmatrix}
1 & 1 & \cdots & 1
\end{bmatrix}^{\top} \in \mathbb{R}^N,\qquad 
\bar{\mathbf{1}} \coloneqq \frac{1}{N} \mathbf{1}^{\top}
\]
and write $\ave(v) \coloneqq \bar{\mathbf{1}}v $ for $v \in \mathbb{R}^N$.
The graph Laplacian of an undirected graph $G$ is denoted by $L(G)$. 
We denote the Lambert $W$-function by $W(y)$ 
for $y \geq 0$.
In other words, $W(y)$ is the solution $x \geq 0$ of 
the transcendental equation $x e^x = y$.
Throughout this paper, we shall use  the following fact frequently 
without comment:
For $a,\omega > 0$ and $c \in \mathbb{R}$,
the solution $x=x^*$ of the transcendental equation 
$
a(x-c) = e^{-\omega x}
$
can be written as
\[
x^* = \frac{1}{\omega} W\left(
\frac{\omega e^{-\omega c}}{a}
\right) + c.
\]

\section{System Model}
\label{sec:sys_model}
\subsection{Multi-agent system}
\label{sec:MAS}
Let $N \in \mathbb{N}$ be $N \geq 2$, and
consider a multi-agent system with $N$ agents.
Each agent has a label $i \in \mathcal{N} \coloneqq 
\{1,2,\dots,N \}$.
For every $i \in \mathcal{N}$,
the dynamics of agent~$i$ is given by
\begin{equation}
\label{eq:each_agent_dynamics}
\dot x_i(t) = u_i(t),\quad t \geq 0; \qquad x_i(0) = x_{i0} \in \mathbb{R},
\end{equation}
where $x_i(t) \in \mathbb{R}$ and $u_i(t) \in \mathbb{R}$
are the state and the control input of agent~$i$, respectively.
The network topology of the multi-agent system
is given by a fixed undirected graph $G = (\mathcal{V}, \mathcal{E})$
with vertex set $\mathcal{V} = \{v_1,v_2,\dots,v_N \}$ and 
edge set 
\[\mathcal{E} \subseteq \{
(v_i,v_j) \in \mathcal{V} \times \mathcal{V}: i\not=j
\}.
\]
If $(v_i,v_j) \in \mathcal{E}$, then
agent~$j$ is called a \textit{neighbor} of agent~$i$, and
these two agents can measure the relative states and communicate with each other.
For $i \in \mathcal{N}$,
we denote by $\mathcal{N}_i$ the set of all neighbors of agent~$i$ and
by $d_i$ the degree of the node $v_i$, that is, the cardinality of the set $\mathcal{N}_i$.

Consider the ideal case without 
quantization or self-triggered sampling, and set
\begin{equation}
\label{eq:ideal_control_input}
u_i(t) = - \sum_{j \in \mathcal{N}_i} \big(x_i(t) - x_j(t) \big)
\end{equation}
for $t \geq 0$ and $i \in \mathcal{N}$.
It is well known that the multi-agent system to which the control input 
\eqref{eq:ideal_control_input} is applied achieves
average consensus under the following assumption.
\begin{assumption}
	\label{assump:connectivity}
	The undirected 
	graph $G$ is connected.
\end{assumption}

In this paper,
we place Assumption \ref{assump:connectivity}.
Moreover, we make two assumptions, which
are used to avoid the saturation of quantization schemes.
These assumptions are relative-state analogues of
the assumptions
in the previous studies  on quantized consensus
based on absolute state measurements (see, e.g., Assumptions~3 and 4 of \cite{Ma2018}).
\begin{assumption}
	\label{assump:initial_bound}
	A bound $E_0>0$ satisfying
	\[
	\left|x_{i0} - \frac{1}{N}\sum_{j \in \mathcal{N}} x_{j0} \right| \leq E_0\qquad 
	\text{for all~} i \in \mathcal{N}
	\] is known by all agents.
\end{assumption}

\begin{assumption}
	\label{assump:ni_bound}
	A bound $\tilde d \in \mathbb{N}$  satisfying
	\[
	d_i \leq \tilde d\qquad 
	\text{for all~} i \in \mathcal{N}
	\]
	is known by all agents.
\end{assumption}

We make an assumption on the number $R$
of quantization levels.
\begin{assumption}
	\label{assump:quantization_number}
	The number $R$ of quantization levels  is an odd number, i.e.,
	$R = 2R_0 + 1$ for some $R_0 \in \mathbb{N}_0$.
\end{assumption}

In this paper,
we study the following notion of consensus of
multi-agent systems under Assumption~\ref{assump:initial_bound}.
\begin{definition}
	\normalfont
	The multi-agent system \textit{achieves consensus exponentially 
		with decay rate $\omega >0$} under 
	Assumption~\ref{assump:initial_bound} if
	there exists a constant $\Omega>0$, independent
	of $E_0$, such that
	\begin{equation}
	\label{eq:consensus}
	|x_i(t) - x_j(t)| \leq \Omega
	E_0 e^{-\omega t} 
	\end{equation}
	for all $t \geq 0$ and $i,j \in \mathcal{N}$.
\end{definition}

\subsection{Quantization scheme}
\label{sec:Quantization_scheme}
Let $E>0$ be a quantization range and let $R \in \mathbb{N}$  be 
the number  of quantization levels satisfying Assumption~\ref{assump:quantization_number}.
We assume that $E$ and $R$ 
are shared 
among all agents. 
We apply uniform quantization to the interval $[-E,E]$. Namely,
a quantization function $Q_{E,R}$
is defined by
\begin{equation*}
Q_{E,R} [z]\coloneqq
\begin{cases}
\dfrac{2pE}{R}  & \text{if $\dfrac{(2p-1)E}{R}  < z \leq  \dfrac{(2p+1)E}{R} $}  \vspace{8pt}\\
0 & \text{if  $-\dfrac{E}{R} \leq z \leq \dfrac{E}{R}$} \vspace{8pt}\\
-Q_{E,R} [-z] &  \text{if $z < -\dfrac{E}{R} $}
\end{cases}
\end{equation*}
for $z \in [-E,E]$,
where $p\in \mathbb{N}$ and $p \leq R_0$.
By construction,
\[
\big |z - Q_{E,R}[z]
\big| \leq \frac{E}{R}
\]
for all $z \in [-E,E]$.
The agents use a  fixed  $R$
but change $E$ in order to 
achieve consensus asymptotically. 
In other words, $E$ is the scaling parameter of the quantization scheme.

Let
$\{t_k^i\}_{k\in \mathbb{N}_0}$ be a strictly increasing sequence with
$t_0^i \coloneqq 0$, and 
$t_k^i$ is the $k$-th sampling time of agent $i$.
To describe the quantized data used at time $t=t_k^i$
for  $k\in \mathbb{N}_0$,
we 
assume for the moment that a certain function $E:[0,\infty)
\to (0,\infty)$ satisfies
the unsaturation condition
\begin{equation}
\label{eq:E_bound}
|x_i(t_k^i) - x_j(t_k^i)| \leq E(t_k^i)
\quad 
\text{for all $j \in \mathcal{N}_i$}.
\end{equation}

Agent~$i$ measures the relative state
$x_i(t_k^i) - x_j(t_k^i) $ 
for each neighbor~$j \in \mathcal{N}_i$ 
and obtains its 
quantized value
\[
q_{ij}(t_k^i) \coloneqq Q_{E(t_k^i),R} [x_i(t_k^i) - x_j(t_k^i)  ].
\]
Then
agent~$i$ sends to each neighbor~$j \in \mathcal{N}_i$ 
the sum
\[
q_i(t_k^i) \coloneqq\sum_{j \in \mathcal{N}_i}q_{ij}(t_k^i).
\] 
The neighbors use the sum $q_i(t_k^i)$ to 
calculate the next sampling time, not the input.
This data transmission implies that
the agents use information not only about direct neighbors but also
about two-hop neighbors as in the self-triggering mechanism developed in \cite{Fan2015}.

The sum $q_i(t_k^i)$ consists of the quantized values, and therefore
agent~$i$ can transmit $q_i(t_k^i)$  without errors
even through finite-capacity channels.
In fact, since $R$ is an odd number under Assumption~\ref{assump:quantization_number}, 
the sum $q_i(t_k^i)$ belongs to the finite set
\[
\left\{
\frac{2pE(t_k^i)}{R}: p \in \mathbb{Z}\text{~~~and~~}-\tilde d R_0 \leq p \leq \tilde dR_0
\right\},
\]
where $\tilde d \in \mathbb{N}$ is as in Assumption~\ref{assump:ni_bound}.
The encoder of agent~$i$ assigns an index to each value $2pE/R$ and
transmits  
the index corresponding to the sum $q_i(t_k^i)$  to
the decoder of each neighbor~$j \in \mathcal{N}_i$.
Since the agents share $E$, $R$, and $\tilde d$,
the decoder can generate the sum $q_i(t_k^i)$ from the received index.

\subsection{Triggering mechanism}
Let a strictly increasing sequence
$\{t_k^i\}_{k\in \mathbb{N}_0}$ with $t_0^i \coloneqq 0$ be
the sampling times 
of agent~$i\in \mathcal{N}$ as in Section~\ref{sec:Quantization_scheme}, and
let $k \in \mathbb{N}_0$.
As in the ideal case \eqref{eq:ideal_control_input},
the control input $u_i(t)$ of agent~$i $ is given by
the sum of the quantized relative state,
\begin{equation}
\label{eq:each_agent_input}
u_i(t) = - q_{i}(t_k^i) = - \sum_{j \in \mathcal{N}_i} Q_{E(t_k^i),R} [x_i(t_k^i) - x_j(t_k^i)  ],
\end{equation}
for $t_k^i \leq t <t_{k+1}^{i}$ when the unsaturation condition
\eqref{eq:E_bound} is satisfied.
Then
the dynamics of agent~$i$ can be written as
\begin{align}
\label{eq:multi_agent_sys_ind}
\dot x_i(t) = - \sum_{j \in \mathcal{N}_i}
\big(x_i(t) - x_j(t)\big) + f_i(t)+g_i(t),
\end{align}
where $f_i(t)$ and $g_i(t)$ are, respectively, the errors due to sampling and 
quantization defined by
\begin{align}
f_i(t) &\coloneqq 
\sum_{j \in \mathcal{N}_i} \big(x_i(t) - x_j(t)\big) - 
\sum_{j \in \mathcal{N}_i} \big(x_i(t_k^i) - x_j(t_k^i)\big) \label{eq:fi_def} \\
g_i(t) &\coloneqq 
\sum_{j \in \mathcal{N}_i} \big(x_i(t_k^i) - x_j(t_k^i)\big) -
q_i(t_k^i)  \label{eq:gi_def} 
\end{align}
for $t_k^i \leq t < t_{k+1}^i$.

We make a triggering condition on the error $f_i$ due to sampling.
From the dynamics~\eqref{eq:each_agent_dynamics} and the input~\eqref{eq:each_agent_input}
of each agent, we have that
for all $t_k^i \leq t < t_{k+1}^i$, 
\begin{align}
x_i(t) - x_i(t_k^i) &= \int_{t_k^i}^{t}
u_i(s) ds =-(t-t_k^i) q_i(t_k^i) \label{eq:xi_diff}\\
x_j(t) - x_j(t_k^i) &=
\int_{t_k^i}^{t}
u_j(s) ds. \label{eq:xj_diff}
\end{align}
Substituting \eqref{eq:xi_diff} and \eqref{eq:xj_diff} into 
\eqref{eq:fi_def} motivates us to consider the following function obtained only from the inputs:
\begin{align}
\label{eq:fik_def}
f_k^i(\tau) 
&\coloneqq
\sum_{j \in \mathcal{N}_i} \int_{t_k^i}^{t_k^i+\tau}
\big( u_i(s) - u_j(s) \big) ds \notag \\
&= 
-\tau d_i q_i(t_k^i)- 
\sum_{j \in \mathcal{N}_i} \int_{t_k^i}^{t_k^i+\tau}
u_j(s) ds
\end{align}
for $\tau \geq 0$. Notice that 
\[
f_i(t_k^i + \tau) = f_k^i(\tau)
\]
for all $\tau \in [0,t_{k+1}^i - t_k^i)$.
Using the quantization range $E(t)$, we define
the $(k+1)$-th sampling time $t_{k+1}^i$
of agent~$i \in \mathcal{N}$ by
\begin{equation}
\label{eq:STM}
\left\{
\begin{alignedat}{4}
t_{k+1}^i &\coloneqq
t_k^i + \min \{\tau_k^i,\,\tau^i_{\max}  \}\\
\tau_k^i &\coloneqq 
\inf\{
\tau \geq \tau^i_{\min} :
|f_k^i(\tau)| \geq \delta_i E(t_k^i+\tau)
\},
\end{alignedat}
\right.
\end{equation}
where $\delta_i>0$ is a threshold and 
$\tau^i_{\max}, \tau^i_{\min} >0$ are upper and lower bounds of
inter-event times, respectively, i.e., $\tau^i_{\min} \leq \tau_k^i
\leq \tau^i_{\max}$.

The behaviors of the errors $f_i(t)$ and $g_i(t)$ 
can be roughly described as follows.
Under the triggering mechanism \eqref{eq:STM}, the error $|f_i(t)|$ due to sampling
is upper-bounded  by $\delta_i E(t)$.
The error $|g_i(t)|$ due to quantization is also 
bounded from above by a constant multiple of $E(t_k^i)$ for 
$t_k^i \leq t <t_{k+1}^{i}$ when
the quantizer is not saturated.
Hence,  if $E(t)$ decreases to zero as $t \to \infty$, then both errors $f_i(t)$ and 
$g_i(t)$ also go to zero.

After some preliminaries in Section~\ref{sec:preliminaries},
Section~\ref{sec:Analysis} is devoted to finding a quantization range $E(t)$, 
a threshold $\delta_i$, and
upper and lower bounds $\tau^i_{\max},\tau^i_{\min} $ of inter-event times such that  
consensus \eqref{eq:consensus} as well as the unsaturation condition 
\eqref{eq:E_bound} are satisfied. 
In Section~\ref{sec:self_triggered_computation}, we present
a method for agent~$i$ to compute 
the sampling times $\{t_k^i\}_{k\in \mathbb{N}_0}$
in a self-triggered fashion.

We conclude this section by making two remarks on the triggering mechanism~\eqref{eq:STM}. First,
the constraint $\tau_k^i  \geq \tau_{\min}^i$ is made solely to simplify the consensus analysis, and
agent $i$ can compute 
the sampling times
$\{t_k^i\}_{k\in \mathbb{N}_0}$
without using the lower bound $\tau^i_{\min}$.
Second,
continuous communication with the neighbors is not required to compute the sampling times, although the inputs of the 
neighbors are used 
in the triggering mechanism~\eqref{eq:STM}. It is enough for agent~$i$ to communicate with the neighbor~$j\in \mathcal{N}_i$
at their sampling times $\{t_k^i\}_{k\in \mathbb{N}_0}$ and $\{t_k^j\}_{k\in \mathbb{N}_0}$. 
In fact,
the inputs are piecewise-constant functions, and
agent~$i$ can know the input $u_j$ of the neighbor~$j $
from the received data $q_j(t_k^j)$. Based on the updated information on $q_j(t_k^j)$,
agent~$i$ recalculates the next sampling time.
We will discuss these issues in detail in Section~\ref{sec:self_triggered_computation}.

\section{Preliminaries}
\label{sec:preliminaries}
In this section, we 
introduce a semi-norm on $\mathbb{R}^N$ 
and basic properties of sampling times.
The reader eager to pursue the consensus analysis 
of multi-agent systems might skip detailed proofs
in this section and return to them when needed.

\subsection{Semi-norm for consensus analysis}

Inspired by the norm used in the theory of operator semigroups (see, e.g., the proof of Theorem~5.2 in Chapter~1 of 
\cite{Pazy1983}), we introduce a new semi-norm 
on $\mathbb{R}^N$, 
which will
lead to the semi-contractivity property \cite{Jafarpour2022,Pasquale2022} of 
the matrix exponential of the negative Laplacian matrix.
\begin{lemma}
	\label{lem:semi_norm}
	Let $\|\cdot\|$ be an arbitrary norm on $\mathbb{R}^N$, and let
	$L,F \in \mathbb{R}^{N\times N}$. Assume that $\Gamma>0$ and
	$\gamma \in \mathbb{R}$ satisfy
	\begin{align}
	\|e^{-Lt}(v - \ave(v) \mathbf{1}) \| &\leq \Gamma e^{-\gamma t} \|Fv\|
	\label{eq:Ptau1_prop}
	\end{align} 
	for all $v\in \mathbb{R}^N$ and $t \geq 0$.
	Then the function $\vertiii{\cdot}\colon \mathbb{R}^N \to [0,\infty)$
	defined by
	\begin{align*}
	\vertiii{v} 
	\coloneqq \sup_{t \geq  0} \|e^{\gamma t}e^{-Lt}
	(v - \ave(v)\mathbf{1})\|,\quad v\in \mathbb{R}^N,
	\end{align*}
	satisfies the following properties:
	\begin{enumerate}
		\renewcommand{\labelenumi}{\alph{enumi})}
		\item 
		For all $v\in \mathbb{R}^N$,
		\[
		\|v - \ave(v)\mathbf{1}\|
		\leq \vertiii{v} \leq  \Gamma \|Fv\|.
		\]
		\item 
		For all $v\in \mathbb{R}^N$,
		$\vertiii{v} = 0$ if and only if
		$v =  \ave(v)\mathbf{1}$.
		\item
		$\vertiii{\cdot}$ is a semi-norm on $\mathbb{R}^N$, i.e., for all
		$v,w \in \mathbb{R}^N$ and $\rho \in \mathbb{R}$,
		\[
		\vertiii{\rho v} = |\rho |~\!\vertiii{v},\quad 
		\vertiii{v+w} \leq \vertiii{v} + \vertiii{w}.
		\]
		\item
		If $L$
		satisfies $\bar{\mathbf{1}} L = 0$
		and $L\mathbf{1} = 0$, then 
		\[
		\vertiii{e^{-Lt}v} \leq e^{-\gamma t} \vertiii{v}
		\] 
		for all $v\in \mathbb{R}^N$ and $t \geq 0$.
	\end{enumerate}
\end{lemma}
\begin{proof}
	Let $v,w\in \mathbb{R}^N$ and $\rho \in \mathbb{R}$ be given.

	a) By definition, we have
	\begin{align*}
	\vertiii{v} 
	\geq \|e^{\gamma 0}e^{-L0}
	(v - \ave(v)\mathbf{1})\| 
	= \|v - \ave(v)\mathbf{1}\|.
	\end{align*}
	The inequality
	\eqref{eq:Ptau1_prop}  yields
	\begin{align*}
	\|e^{\gamma t}e^{-Lt}
	(v - \ave(v)\mathbf{1})\|
	&\leq \Gamma \|Fv\| 
	\end{align*}
	for all
	$t \geq 0$. Hence,
	$
	\vertiii{v} 
	\leq \Gamma \|Fv\|.
	$
	
	b) This follows immediately from the definition of $\vertiii{\cdot}$.
	
	c)  We obtain
	\begin{align*}
	\vertiii{\rho v} 
	= |\rho|
	\sup_{t \geq 0} \|e^{\gamma t}e^{-Lt}
	(v - \ave(v)\mathbf{1})\| 
	= |\rho|~\!\vertiii{v}.
	\end{align*}
	Since
	$
	\ave(v+w) = \ave(v) + \ave(w),
	$
	it follows from
	the triangle inequality for the norm $\|\cdot\|$ that
	\begin{align*}
	\vertiii{v+w} 
	& \leq
	\sup_{t \geq 0} \big(\|e^{\gamma t}e^{-Lt}
	(v - \ave(v)\mathbf{1})\| + \|e^{\gamma t}e^{-Lt}
	(w - \ave(w)\mathbf{1})\| \big) \\
	&  \leq
	\sup_{t \geq 0} \|e^{\gamma t}e^{-Lt}
	(v - \ave(v)\mathbf{1})\| 
	+ 
	\sup_{t \geq 0} \|e^{\gamma t}e^{-Lt}
	(w - \ave(w)\mathbf{1})\| \\
	&  = \vertiii{v}+\vertiii{w}.
	\end{align*}

	d)
	By assumption, 
	\begin{align*}
	&\ave(e^{-Lt}v)\mathbf{1} = 
	(
	\bar{\mathbf{1}} e^{-Lt}v
	)\mathbf{1} = 
	(\bar{\mathbf{1}} v)\mathbf{1} \\
	&\quad=
	\ave(v)\mathbf{1}  =  \ave(v) (e^{-Lt}\mathbf{1}) = 
	e^{-Lt} (\ave(v) \mathbf{1}).
	\end{align*}
	This yields
	\begin{align*}
	\vertiii{e^{-Lt}v} 
	&=
	\sup_{s \geq 0} \|e^{\gamma s}e^{-L(s+t)}
	(v - \ave(v)\mathbf{1})\| \\
	&\leq e^{-\gamma t}
	\sup_{s \geq 0} \|e^{\gamma s}e^{-Ls}
	(v - \ave(v)\mathbf{1})\| \\
	&=  e^{-\gamma t} \vertiii{v}
	\end{align*}
	for all $t \geq 0$.
\end{proof}
\begin{remark}
	\normalfont
	If the inequality in 
	Lemma~\ref{lem:semi_norm}.d) is satisfied for some $\gamma >0$,
	then $e^{-Lt}$ is a semi-contraction with respect to the semi-norm $\vertiii{\cdot}$ 
	for all $t >0$.
	In Lemma~9 of \cite{Jafarpour2022}, a
	more general method is presented for constructing
	such semi-norms.
	The tuning parameter of this method is a matrix whose 
	kernel coincides with the span $\{\alpha \mathbf{1} : \alpha \in\mathbb{R}\}$.
	Since the constants $\Gamma$ and $\gamma$ in \eqref{eq:Ptau1_prop}
	are easier to tune for the joint design of 
	a quantizer and a self-triggering mechanism, 
	we will use
	Lemma~\ref{lem:semi_norm} in the consensus analysis.
\end{remark}

\subsection{Basic properties of sampling times}
\label{sec:sampling}
Let $\{t_k^i \}_{k \in \mathbb{N}_0}$ be a strictly increasing sequence
of real numbers with $t_0^i \coloneqq 0$
for $i \in \mathcal{N} = \{1,2,\dots,N\}$.
Set $t_0 \coloneqq 0$ and $k_i(0) \coloneqq 0$ for $i \in \mathcal{N}$. Define 
\begin{align*}
t_{\ell+1} &\coloneqq \min_{i \in \mathcal{N}} t^i_{k_i(\ell) + 1} \\
k_i(\ell +1) &\coloneqq 
\max\big\{
k \in \mathbb{N}_0 : k \leq k_i(\ell) +1 \text{~~and~~} t_k^i \in \{
t_ 0,t_1,\dots,t_{\ell+1}
\}
\big\}
\end{align*}
for $\ell \in \mathbb{N}_0$ and $i \in \mathcal{N}$.
Roughly speaking, 
in the context of the multi-agent system,
$\{t_\ell\}_{\ell \in \mathbb{N}_0}$ are all sampling times of the agents
without duplication, and
$k_i(\ell)$ is the number of times agent~$i$ has measured the relative states
on the interval 
$(0,t_{\ell}]$. 
Hence $t^i_{k_i(\ell)}$ is the latest sampling time of agent~$i$
at time $t=t_{\ell}$.
Define $\mathcal{I}(0) \coloneqq \mathcal{N}$ and
\[
\mathcal{I}(\ell+1) \coloneqq 
\{
i \in \mathcal{N} :  t_{\ell+1} = t_{k_i(\ell)+1}^i
\}
\]
for $\ell \in \mathbb{N}_0$.
In our multi-agent setting,
$\mathcal{I}(\ell)$ represents the set of agents
measuring the relative states at $t= t_{\ell}$.

\begin{proposition}
	\label{prop:basic_property}
	Let $\{t_k^i \}_{k \in \mathbb{N}_0}$
	be a strictly increasing sequence 
	of real numbers with $t_0^i \coloneqq 0$
	for $i \in \mathcal{N} = \{
	1,2,\dots,N\}$.
	The sequences
	$\{t_\ell\}_{\ell \in \mathbb{N}_0}$ and $\{k_i(\ell)\}_{\ell \in \mathbb{N}_0}$ defined as above have the following properties
	for all $\ell \in \mathbb{N}_0$ and $i \in \mathcal{N}$:
	\begin{enumerate}
		\renewcommand{\labelenumi}{\alph{enumi})}
		\item $k_i(\ell) \leq k_i(\ell+1) \leq k_i(\ell) + 1$.
		\item 
		$k_i(\ell+1) = k_i(\ell) +1$ 
		if and only if $i \in \mathcal{I}(\ell+1)$. In this case,
		\[
		t_{\ell+1} = t_{k_i(\ell+1)}^i.
		\]
		\item $t^i_{k_i(\ell)} \leq  t_\ell < t_{\ell+1}$. 
		\item If $t_k^i \leq t_{\ell_1}$ for some $k,\ell_1 \in \mathbb{N}_0$, 
		then there exists $\ell_0 \in \mathbb{N}_0$ with $\ell_0 \leq \ell_1$
		such that $ t_k^i = t_{\ell_0}$.
		\item If $
		t^i_{k+1} - t^i_k \leq \tau^i_{\max}
		$
		for all $k \in \mathbb{N}_0$,
		then
		\[
		t_{\ell+1} \leq t^i_{k_i(\ell)} + \tau^i_{\max}.
		\]
		\item 
		If for all $i \in \mathcal{N}$, there exists $\tau_{\min}^i >0$
		such that 
		\[
		t_{k+1}^i - t_k^i \geq \tau_{\min}^i,
		\] then
		$t_{\ell} \to \infty$ as $\ell \to \infty$.
	\end{enumerate}
\end{proposition}
\begin{proof}
	a) The inequality 
	\[
	k_i(\ell+1) \leq k_i(\ell) + 1
	\] 
	follows immediately
	from the definition of $k_i(\ell+1)$.
	Since $k_i(0) = 0 \leq k_i(1)$, the inequality 
	\begin{equation}
	\label{eq:ki_ell}
	k_i(\ell) \leq k_i(\ell+1)
	\end{equation}
	holds for $\ell = 0$. Suppose that 
	the inequality \eqref{eq:ki_ell} holds for some $\ell \in \mathbb{N}_0$.
	Then 
	\begin{align*}
	&\big\{
	k \in \mathbb{N}_0 : 
	k \leq k_i(\ell) +1\text{~~and~~}
	t_k^i \in \{
	t_0,t_1,\dots,t_{\ell+1}
	\}
	\big\} \\
	&\hspace{3pt} \subseteq
	\big\{
	k \in \mathbb{N}_0:k \leq k_i(\ell+1) +1\text{~~and~~} t_k^i \in \{
	t_0,t_1,\dots,t_{\ell+2}
	\}
	\big\},
	\end{align*}
	which yields $k_i(\ell+1) \leq k_i(\ell+2)$.
	Therefore, the inequality \eqref{eq:ki_ell} 
	holds for all $\ell \in \mathbb{N}_0$ by induction.
	
	b)  Assume that $k_i(\ell+1) = k_i(\ell) +1$.
	By the definition of $t_{\ell+1}$, 
	we obtain $t_{\ell + 1} \leq  t^i_{k_i(\ell) + 1}$.
	On the other hand, 
	$t^i_{k_i(\ell+1)} \in \{t_0,\dots,t_{\ell+1} \}$ by
	the definition of $k_i(\ell+1)$.
	Since $\{t_\ell\}_{\ell \in \mathbb{N}_0}$ is a nondecreasing sequence by a), 
	it follows that
	\[
	t^i_{k_i(\ell) + 1} = t^i_{k_i(\ell+1)} \leq t_{\ell+1}.
	\]
	Hence 
	\[
	t_{\ell+1} = t^i_{k_i(\ell) + 1} = t^i_{k_i(\ell+ 1) }.
	\]
	
	Conversely, assume that $t_{\ell+1} = t^i_{k_i(\ell) + 1}$.
	Then 
	\[
	t^i_{k_i(\ell) + 1} \in \{
	t_0,t_1,\dots,t_{\ell+1}
	\}.
	\] 
	By the definition of $k_i(\ell+1)$, we obtain $k_i(\ell + 1) = k_i(\ell) + 1$.

	c)
	The definition of $k_i(\ell)$ directly yields
	\[
	t^i_{k_i(\ell)} \leq  t_{\ell}.
	\]	
	It remains to show that 
	\[
	t_{\ell} < t_{\ell+1}.
	\]
	By construction, 
	$\mathcal{I}(\ell) \not= \emptyset$ holds for all
	$\ell \in \mathbb{N}_0$.
	First, we consider the case
	\[
	\mathcal{I}(\ell) \cap \mathcal{I}(\ell+1) = \emptyset.
	\]
	Since $\mathcal{I}(0) = \mathcal{N}$
	by definition, we obtain $\ell \geq 1$.
	Let $i \in \mathcal{I}(\ell+1)$. 
	Then
	$t_{\ell+1} = t^i_{k_i(\ell)+1}$.
	On the other hand, $i \not \in \mathcal{I}(\ell)$
	and hence 
	\[
	t_{\ell} < t^i_{k_i(\ell-1)+1}.
	\] 
	Since 
	$k_i(\ell-1) = k_i(\ell)$ by a) and b), 
	we obtain
	\[
	t_{\ell} < t^i_{k_i(\ell-1)+1} = t^i_{k_i(\ell)+1} =  t_{\ell+1}.
	\]
	
	Next, assume that 
	\[
	\mathcal{I}(\ell) \cap \mathcal{I}(\ell+1) 
	\not= \emptyset.
	\]
	Let 
	\[
	i \in \mathcal{I}(\ell) \cap \mathcal{I}(\ell+1) .
	\]
	Then $t_{\ell+1} = t^i_{k_i(\ell)+1} $ from
	$i \in \mathcal{I}(\ell+1)$.
	If $\ell =0$, then 
	\[
	t_{\ell} = t_0 = 0 = t^i_{k_i(0)} = t^i_{k_i(\ell)}.
	\]
	If $\ell \geq 1$, then we have from $i \in \mathcal{I}(\ell)$ and b)
	that 
	\[
	t_{\ell}  = t^i_{k_i(\ell)}.
	\]
	Since $t^i_{k_i(\ell)} <  t^i_{k_i(\ell)+1}$, 
	it follows that $t_{\ell} < t_{\ell+1}$.

	d) We have from c) that
	\[
	t_{\ell_1} < t_{\ell_1+1} \leq t^i_{k_i(\ell_1)+1}.
	\]
	This and the assumption $t_k^i \leq t_{\ell_1}$ yield
	$t_k^i < t^i_{k_i(\ell_1)+1}$, and therefore $k \leq k_i(\ell_1)$. 
	Let 
	\[
	\ell_0 \coloneqq \min\{\ell \in \mathbb{N}_0: \ell \leq \ell_1
	\text{~~and~~}k = k_i(\ell) \}.
	\]
	If $\ell_0=0$, then we obtain
	\[
	t_k^i = t_0^i = 0 = t_0 = t_{\ell_0}.
	\]
	Assume that $\ell_0 \not=0$. Then $k =  k_i(\ell_0)\geq 1$ and 
	\[
	k_i(\ell_0) = k_i(\ell_0-1) + 1.
	\] 
	This and b) yield
	\[
	t^i_k= t_{k_i(\ell_0)}^i = t_{\ell_0}.
	\]

	e)
	Since
	$
	t_{\ell+1} \leq t^i_{k_i(\ell)+1}
	$
	by the definition of $t_{\ell+1}$,
	it follows that
	\[
	t_{\ell + 1} -  t^i_{k_i(\ell)} \leq t^i_{k_i(\ell)+1} -  t^i_{k_i(\ell)} \leq 
	\tau^i_{\max}.
	\]

	f) For all $\ell\in \mathbb{N}_0$,
	there exists $i \in \mathcal{N}$
	such that $t_{\ell} \in \{t^i_{k} \}_{k \in \mathbb{N}_0}$.
	We have from c) that $t_{\ell} \not= t_{\ell+1} $.
	Since  $\mathcal{N}$ is a set with finite elements,
	there exist $i \in \mathcal{N}$ and 
	a 
	subsequence $\{t_{\ell(p)} \}_{p \in \mathbb{N}_0}$
	of $\{t_{\ell} \}_{\ell \in \mathbb{N}_0}$
	such that 
	\[
	t_{\ell(p)} \in  \{t^i_{k} \}_{k \in \mathbb{N}_0}
	\]
	for all $p \in \mathbb{N}_0$. For each $p \in \mathbb{N}_0$, let 
	$k(p) \in \mathbb{N}_0$ satisfy 
	\[
	t_{\ell(p)} = t^i_{k(p)}.
	\]
	Assume, to get a contradiction, that $\sup_{\ell \in \mathbb{N}_0}t_\ell 
	< \infty$. Take 
	\[
	0 < \varepsilon < \tau_{\min}^i.
	\]
	There exists $p_0 \in \mathbb{N}_0$ such that
	\[
	t_{\ell(p+1)} - t_{\ell(p)} < \varepsilon
	\]
	for all $p \geq p_0$.
	
	Choose $p \geq p_0$ arbitrarily. We obtain
	\[
	t_{\ell(p)} = t^i_{k(p)} < t^i_{k(p)+1} \leq t^i_{k(p+1)} = t_{\ell(p+1)} .
	\]
	Hence
	\begin{equation}
	\label{eq:kp1+kp_bound}
	t^i_{k(p)+1} - t^i_{k(p)} \leq
	t_{\ell(p+1)} - t_{\ell(p)} < \varepsilon.
	\end{equation}
	By assumption,
	\[
	t^i_{k(p)+1} - t^i_{k(p)} \geq \tau_{\min}^i > \varepsilon,
	\]
	which contradicts \eqref{eq:kp1+kp_bound}.
\end{proof}

\section{Consensus Analysis }
\label{sec:Analysis}
In this section, first we define a semi-norm based on the maximum norm.
Next, we obtain a bound of the state with respect to the semi-norm
for the design of the quantization range.
After these preparations, we give a sufficient condition for consensus 
in the main theorem. 
Finally, we find bounds of the constant $\Gamma$ in \eqref{eq:Ptau1_prop} corresponding to our multi-agent setting.

Throughout this and the next sections, 
we consider the quantized self-triggered multi-agent system presented in 
Section~\ref{sec:sys_model}.
Let
$\{t_k^i\}_{k\in \mathbb{N}_0}$ with $t_0^i \coloneqq 0$ be
the sampling times 
of agent~$i \in \mathcal{N}$, which are given in \eqref{eq:STM}.
Define $\{t_\ell\}_{\ell \in \mathbb{N}_0}$ and $\{k_i(\ell)\}_{\ell \in \mathbb{N}_0}$ as in Section~\ref{sec:sampling}.
We let $L \coloneqq L(G)$, where
$G$ is the undirected graph of the multi-agent system.

Define 
\begin{align*}
x(t) &\coloneqq 
\begin{bmatrix}
x_1(t) & x_2(t) & \cdots & x_n(t)
\end{bmatrix}^{\top} \\
x_0 &\coloneqq 
\begin{bmatrix}
x_{10} & x_{20} & \cdots & x_{n0}
\end{bmatrix}^{\top}\\
f(t) &\coloneqq 
\begin{bmatrix}
f_1(t) & f_2(t) & \cdots & f_n(t)
\end{bmatrix}^{\top} \\
g(t) &\coloneqq 
\begin{bmatrix}
g_1(t) & g_2(t) & \cdots& g_n(t)
\end{bmatrix}^{\top}
\end{align*}
for $t \geq 0$.
Then we have from the dynamics \eqref{eq:multi_agent_sys_ind} of individual agents  that
\begin{equation*}
\Sigma_{\mathrm{MAS}}\quad
\begin{cases}
\dot x(t) = -Lx(t) + f(t)+g(t),\quad t  \geq 0;\\
x(0) = x_0.
\end{cases}
\end{equation*}

\subsection{Semi-norm based on the maximum norm}
We start by showing the following simple result.
\begin{lemma}
	\label{lem:Ptau_bound}
	Let $N \in \mathbb{N}$ satisfy $N \geq 2$ and let 
	$L$ be the Laplacian matrix of  
	a connected  undirected graph with $N$ vertices.
	Let 
	$\|\cdot\|$ be an arbitrary norm on $\mathbb{R}^N$ and the corresponding 
	induced norm on $\mathbb{R}^{N \times N}$. 
	Fix $\gamma \leq \lambda_2(L)$, and define
	\begin{equation*}
	\Gamma \coloneqq \sup_{t \geq 0} \|e^{\gamma t }(e^{-Lt} - 
	\mathbf{1}\bar{\mathbf{1}})\|.
	\end{equation*}
	Then $\Gamma < \infty$ and the inequalities
	\begin{align}
	\label{eq:M_rho}
	\|e^{-Lt}(v - \ave(v)\mathbf{1})\|
	&\leq \Gamma e^{-\gamma t } \|v - \ave(v)\mathbf{1}\|\\
	\label{eq:M_rho2}
	\|e^{-Lt}(v - \ave(v)\mathbf{1})\|
	&\leq \Gamma e^{-\gamma t } \|v\|
	\end{align}
	hold for all $v\in \mathbb{R}^N$ and $t \geq 0$.
\end{lemma}

\begin{proof}
	Let $\lambda_1,\lambda_2,\lambda_3,\dots,\lambda_N$ be the eigenvalues of $L$. Since the undirected graph corresponding to $L$ is connected, we have that 
	$0$ is an eigenvalue of $L$ with algebraic multiplicity $1$. 
	Let $\lambda_1 \coloneqq 0$ and define
	\begin{align*}
	\Lambda_0 &\coloneqq 
	\diag \big(0 ,\,\lambda_2,\,\lambda_3,\,\cdots,\,\lambda_N \big)\\
	\Lambda &\coloneqq 
	\diag \big(\lambda_2\, ,\lambda_3\,,\cdots,\,\lambda_N \big).
	\end{align*}
	There exists an orthogonal matrix $V_0 \in \mathbb{R}^{N\times N}$ such that 
	\[
	L = V_0 \Lambda_0  V_0^{\top}.
	\]
	Since 
	$\mathbf{1}$ is the eigenvector corresponding to the
	eigenvalue $\lambda_1 = 0$,
	one can decompose $V_0$  into
	\[
	V_0 = 
	\begin{bmatrix}
	\dfrac{\mathbf{1} }{\sqrt{N}} & V
	\end{bmatrix}
	\]
	for some $V \in \mathbb{R}^{N \times(N-1)}$.
	
	Let $v \in \mathbb{R}^N$
	and $t \geq 0$. 
	Noting that 
	\[
	\frac{\mathbf{1} }{\sqrt{N}}
	\left(
	\dfrac{\mathbf{1}^{\top}}{\sqrt{N}}v
	\right) =
	\ave(v) \mathbf{1},
	\]
	we obtain
	\begin{align}
	e^{-Lt}v &= V_0 e^{-\Lambda_0 t}  V_0^{\top}v 
	=
	\ave(v) \mathbf{1} + Ve^{-\Lambda t} V^{\top} v.
	\label{eq:exp_v}
	\end{align}
	Since \[
	\ave\big((\ave(v) \mathbf{1}) \big)\mathbf{1} = \ave(v) \mathbf{1},
	\]
	it follows that
	\begin{align}
	\label{eq:exp_ave_v}
	e^{-Lt}(\ave(v) \mathbf{1}) = 
	\ave(v) \mathbf{1} + Ve^{-\Lambda t} V^{\top} (\ave(v) \mathbf{1}).
	\end{align}
	By \eqref{eq:exp_v} and \eqref{eq:exp_ave_v},
	\begin{equation}
	\label{eq:Ptau_decomp1}
	e^{-Lt}(v - \ave(v) \mathbf{1})  =
	V e^{-\Lambda t}
	V^{\top} (v - \ave(v) \mathbf{1}).
	\end{equation}
	On the other hand,
	using $e^{-Lt} \mathbf{1} =  \mathbf{1}$, we obtain
	\begin{equation}
	\label{eq:exp_aveone}
	e^{-Lt}(\ave(v) \mathbf{1}) = \ave(v) \mathbf{1}.
	\end{equation}
	By \eqref{eq:exp_v} and \eqref{eq:exp_aveone},
	\begin{equation}
	\label{eq:Ptau_decomp2}
	e^{-Lt}(v - \ave(v) \mathbf{1})  =
	V e^{-\Lambda t} V^{\top}
	v.
	\end{equation}
	
	Since
	$\lambda_i \geq \lambda_2(L) \geq \gamma$ 
	for all $i=2,3,\dots,N$,
	it follows that
	\[
	C \coloneqq \sup_{t \geq 0}\| e^{\gamma t}e^{-\Lambda t} \| < \infty.
	\]
	Moreover, \eqref{eq:exp_aveone} gives
	\begin{equation}
	\label{eq:v_compact}
	e^{-Lt}(v - \ave(v) \mathbf{1}) = 
	(e^{-Lt}-  \mathbf{1}\bar{\mathbf{1}})v.
	\end{equation}
	Using \eqref{eq:Ptau_decomp2} and \eqref{eq:v_compact},
	we have 
	\begin{equation}
	\label{eq:Gamma_bound}
	\Gamma
	=
	\sup_{t \geq 0}
	\|V e^{\gamma t}e^{-\Lambda t} V^{\top} \| \leq C \|V\| ~\! \|V^{\top}\| < \infty.
	\end{equation}
	The inequalities \eqref{eq:M_rho} and \eqref{eq:M_rho2}	
	follow from \eqref{eq:Ptau_decomp1} and
	\eqref{eq:Ptau_decomp2}, respectively.
\end{proof}

Fix
a constant
$0 < \gamma \leq \lambda_2(L)$. Here we apply Lemmas~\ref{lem:semi_norm} and \ref{lem:Ptau_bound}
in the case $\|\cdot\| = 
\|\cdot\|_{\infty}$. By Lemma~\ref{lem:Ptau_bound},
\begin{equation}
\label{eq:Gamma_inf_def}
\Gamma_{\infty} \coloneqq
\Gamma_{\infty}(\gamma) \coloneqq \sup_{t \geq 0} \|e^{\gamma t}(e^{-Lt} - 
\mathbf{1}\bar{\mathbf{1}})\|_{\infty} < \infty.
\end{equation}
It is immediate that
\begin{equation}
\label{eq:Gamma_inf_bound}
\Gamma_{\infty} \geq \|e^{\gamma 0}(e^{-L0} - 
\mathbf{1}\bar{\mathbf{1}})\|_{\infty}   = \|I - \mathbf{1}\bar{\mathbf{1}}\|_{\infty}= 2-\frac{2}{N} \geq 1
\end{equation}
for all $N \geq 2$.
We also have
\[
\|e^{-Lt}(v - \ave(v) \mathbf{1}) \|_{\infty} \leq 
\Gamma_{\infty} e^{-\gamma t} \|Fv\|_{\infty},
\]
where $F = I - \mathbf{1} \bar{\mathbf{1}}$ from \eqref{eq:M_rho}
and $F = I$ from \eqref{eq:M_rho2}.
Define 
\begin{equation}
\label{eq:vertiii_inf}
\vertiii{v}_{\infty} \coloneqq 
\sup_{t \geq 0} \|e^{\gamma t}e^{-Lt} 
(v - \ave(v)\mathbf{1})\|_{\infty},\quad v \in \mathbb{R}^N.
\end{equation}
Then $\vertiii{\cdot}_{\infty}$ is a semi-norm on $\mathbb{R}^N$
and satisfies the properties in Lemma~\ref{lem:semi_norm}.
The next lemma motivates us to investigate
the semi-norm of the state $x$ of  $\Sigma_{\mathrm{MAS}}$.
\begin{lemma}
	\label{lem:vivj_bound}
	Define the semi-norm $\vertiii{\cdot}_{\infty}$ as in \eqref{eq:vertiii_inf}.
	Let $v_i\in \mathbb{R}$ be the $i$-th element of $v\in \mathbb{R}^N$ for $i=1,\dots,N$.
	Then
	\begin{equation}
	\label{eq:vivj_bound}
	|v_i -v_j| \leq 2 \vertiii{v}_{\infty}
	\end{equation}
	for all $i,j =1,\dots,N$.
\end{lemma}
\begin{proof}
	For all $i,j = 1,\dots,N$,
	\[
	|v_i -v_j| \leq 
	\left|
	v_i - \ave(v)
	\right| + 
	\left|
	v_j - \ave(v)
	\right| \leq 2\|v - \ave(v)\mathbf{1}\|_{\infty}.
	\]
	By Lemma~\ref{lem:semi_norm}.a), we obtain
	\[
	\|v - \ave(v)\mathbf{1}\|_{\infty} \leq \vertiii{v}_{\infty}.
	\]
	Hence the desired inequality \eqref{eq:vivj_bound}
	holds for all $i,j = 1,\dots,N$.
\end{proof}

\subsection{Design of quantization ranges}
For a given $\omega >0$, the quantization range $E(t)$ is defined by
\begin{align}
\label{eq:E_def}
E(t) \coloneqq 
2\Gamma_{\infty} E_0 e^{-\omega t},\quad t \geq 0.
\end{align}
We also set
\begin{align*}
\kappa(\omega) &\coloneqq \max\left\{
\delta_i +\frac{d_i e^{\omega \tau^i_{\max}}  }{R}: i \in \mathcal{N}
\right\}
\end{align*}
and
\begin{align}
\tilde \tau_{\min}^i &\coloneqq
\min\left\{
\tau > 0 : \tau  
\left(d_i^2 +\sum_{j \in \mathcal{N}_i} d_j e^{\omega 
	\tau_{\max}^j}
\right)  = \delta_i e^{-\omega \tau}
\right\}  \label{eq:tau_min_def_usemin}\\
& =
\frac{1}{\omega} W\left(
\frac{\omega\delta_i}{ d_i^2 +\sum_{j \in \mathcal{N}_i} d_j e^{\omega 
		\tau_{\max}^j}}
\right).\label{eq:tau_min_def}
\end{align}
The following lemma shows that 
$\vertiii{x(t)}_{\infty}$ is bounded by $E(t)/2$
for a suitable decay parameter $\omega$.
\begin{lemma}
	\label{lem:fundamental_lemma}
	Suppose that Assumptions 
	\ref{assump:connectivity}--\ref{assump:quantization_number} hold.
	For each $i \in \mathcal{N}$,
	let
	the lower bound $\tau_{\min}^i$ of inter-event times satisfy
	\[
	0< \tau_{\min}^i \leq 
	\min \{ \tilde \tau_{\min}^i, \, \tau_{\max}^i\}.
	\]
	Assume that
	\begin{equation}
	\label{eq:omega_cond}
	0 < \omega \leq \gamma -2\Gamma_{\infty}\kappa(\omega) ,
	\end{equation}
	and define the quantization range $E(t)$ by \eqref{eq:E_def}.
	Then the state $x$ of  $\Sigma_{\mathrm{MAS}}$ satisfies
	\[
	\vertiii{x(t)}_{\infty} \leq \frac{E(t)}{2}
	\]
	for all $t \geq 0$, where 
	the semi-norm $\vertiii{\cdot}_{\infty}$ is defined by \eqref{eq:vertiii_inf}.
\end{lemma}

\begin{proof}
	Since 
	$
	t_{\ell} \to \infty$ as $\ell \to \infty$
	by Proposition~\ref{prop:basic_property}.f),
	it suffices to prove that 
	\begin{equation}
	\label{eq:x_bound}
	\vertiii{x(t)}_{\infty} \leq \frac{E(t)}{2},\quad 0 \leq t \leq t_{\ell}
	\end{equation}
	for all $\ell \in \mathbb{N}_0$.
	Lemma~\ref{lem:semi_norm}.a) with $F =  I - \mathbf{1} \bar{\mathbf{1}}$ gives
	\[
	\vertiii{x(0)}_{\infty} \leq \Gamma_{\infty} 
	\big\|x(0) - \ave\big(x(0) \big)
	\mathbf{1}\big\|_{\infty}.
	\]
	By Assumption \ref{assump:initial_bound}, we obtain
	\[
	\big\|x(0) - \ave\big(x(0) \big)
	\mathbf{1}\big\|_{\infty} \leq E_0.
	\]
	Since $E(0) = 2\Gamma_{\infty}E_0$ by definition, it follows that
	\[
	\vertiii{x(0)}_{\infty} \leq \frac{E(0)}{2}.
	\]
	Therefore, \eqref{eq:x_bound} holds in the case $\ell = 0$.
	
	We now proceed by induction and assume the inequality
	\eqref{eq:x_bound} to be true for some
	$\ell \in \mathbb{N}_0$.
	Since
	\[
	t_{k_i(p)}^i \leq t_{\ell}
	\]
	for all $p=0,1,\dots,\ell$ and $i \in \mathcal{N}$,
	Lemma~\ref{lem:vivj_bound} yields
	\begin{align*}
	|x_i(t_{k_i(p)}^i) - x_j(t_{k_i(p)}^i)| 
	&\leq  E(t_{k_i(p)}^i)
	\end{align*}
	for all $p = 0,1,\dots,\ell$ and $i,j \in \mathcal{N}$. 
	In other words, the unsaturation condition \eqref{eq:E_bound}
	is satisfied for all $i \in \mathcal{N}$ until $t = t_{\ell}$.

	Fix $i \in \mathcal{N}$.
	Recall that the dynamics 
	of agent~$i$
	is given by \eqref{eq:multi_agent_sys_ind}.
	First we show that the error $f_i$ due to sampling,
	which is defined by \eqref{eq:fi_def}, satisfies 
	\begin{equation}
	\label{eq:fi_bound}
	|f_i(t)| < \delta_i E(t)
	\end{equation}
	for all $t \in [t_{\ell},  t_{\ell+1})$.
	Suppose that $t \in [t_{\ell},t_{\ell+1})$ satisfies
	\[
	t \geq t_{k_i(\ell)}^i + \tau_{\min}^i.
	\]
	Since $t_{k_i(\ell)}^i\leq t_{\ell}$ and $t_{\ell+1}\leq t_{k_i(\ell)+1}^i$
	by definition,
	it follows that
	\[
	f_i(t) = f^i_{k_i(\ell)} (t-t_{k_i(\ell)}^i),
	\]
	where $f^i_k$ is defined by
	\eqref{eq:fik_def}.
	The triggering mechanism \eqref{eq:STM} guarantees that
	\begin{equation*}
	|f^i_{k_i(\ell)} (t-t_{k_i(\ell)}^i)| < \delta_i E(t),
	\end{equation*} 
	and hence \eqref{eq:fi_bound} holds when $t \geq t_{k_i(\ell)}^i + \tau_{\min}^i$.
	
	Let us consider the case where $t \in [t_{\ell},t_{\ell+1})$ satisfies
	\[
	t< t_{k_i(\ell)}^i + \tau_{\min}^i.
	\]
	By definition, $t_{k_i(\ell)}^i \leq t_{\ell}$. Therefore,
	Proposition~\ref{prop:basic_property}.d) yields
	\[
	t_{k_i(\ell)}^i = t_{\ell_0}
	\]
	for some $\ell_0 \in \mathbb{N}_0$ with
	$\ell_0 \leq \ell$. Since
	the unsaturation condition \eqref{eq:E_bound}
	is satisfied until $t = t_{\ell}$, the equations \eqref{eq:xi_diff} and \eqref{eq:xj_diff}
	yield
	\begin{align}
	\label{eq:fik_detail}
	f_i(t) = 
	-(t-t_{\ell_0})d_i q_i(t_{\ell_0}) + (t-t_{\ell}) \sum_{j \in \mathcal{N}_i} q_j(t^j_{k_j(\ell)})  + 
	\sum_{p=0}^{\ell-\ell_0-1}
	(t_{\ell_0+p+1}-t_{\ell_0 + p}) \sum_{j \in \mathcal{N}_i} q_j(t^j_{k_j(\ell_0+p)}).
	\end{align}
	By definition,
	\begin{equation}
	\label{eq:qi_bound}
	|q_i(t_{\ell_0})| \leq d_iE(t_{\ell_0}).
	\end{equation}
	For each $p=0,1,\dots,\ell-\ell_0$ and $j \in \mathcal{N}_i$,
	Proposition~\ref{prop:basic_property}.a), c), and e) give
	\[
	t_{\ell_0} - \tau_{\max}^j \stackrel{\textrm{c)}}{<} 
	t_{\ell_0+1} - \tau_{\max}^j \stackrel{\textrm{e)}}{\leq} 
	t^j_{k_j(\ell_0)}
	\stackrel{\textrm{a)}}{\leq} 
	t^j_{k_j(\ell_0+p)}
	\]
	and hence
	\begin{equation}
	\label{eq:qj_bound}
	|q_j(t^j_{k_j(\ell_0 + p)}) | \leq d_jE(t^j_{k_j(\ell_0+ p) }) \leq 
	d_j e^{\omega 
		\tau_{\max}^j} E(t_{\ell_0}).
	\end{equation}
	Combining \eqref{eq:fik_detail} with the inequalities
	\eqref{eq:qi_bound} and \eqref{eq:qj_bound}, we obtain
	\begin{align*}
	|f_i(t)| \leq 
	(t-t_{\ell_0})  \left(d_i^2 +\sum_{j \in \mathcal{N}_i} d_j e^{\omega 
		\tau_{\max}^j}
	\right) E(t_{\ell_0}).
	\end{align*}
	Since $t-t_{\ell_0} <\tilde \tau_{\min}^i$,
	we see from
	the definition \eqref{eq:tau_min_def_usemin} of $\tilde \tau_{\min}^i$ that
	\begin{align*}
	(t-t_{\ell_0})  
	\left(d_i^2 +\sum_{j \in \mathcal{N}_i} d_j e^{\omega 
		\tau_{\max}^j}
	\right) E(t_{\ell_0})  <\delta_i e^{-\omega(t-t_{\ell_0}) } 
	E(t_{\ell_0}) = \delta_i E(t).
	\end{align*}
	Hence,
	the inequality \eqref{eq:fi_bound} holds also when
	$t< t_{k_i(\ell)}^i + \tau_{\min}^i$.

	Next we study $|g_i(t)|$ for $t_{\ell} \leq t < t_{\ell+1}$,
	where $g_i$ is 
	defined as in \eqref{eq:gi_def} and  is the error due to quantization.
	Since
	the unsaturation condition \eqref{eq:E_bound}
	is satisfied until $t = t_{\ell}$, we have that
	\[
	\big| \big(x_i(t_{k_i(\ell)}^i) - x_j(t_{k_i(\ell)}^i)\big) -
	q_{ij}(t_{k_i(\ell)}^i) \big| \leq \frac{E(t_{k_i(\ell)}^i)}{R} 
	\]
	for all $j \in \mathcal{N}_i$.
	Proposition~\ref{prop:basic_property}.e) shows that 
	$t_{\ell+1} - t_{k_i(\ell)}^i  \leq \tau_{\max}^i$, which gives
	\[
	\frac{E(t_{k_i(\ell)}^i)}{R}  = 
	\frac{e^{\omega(t-t_{k_i(\ell)}^i)}E(t)}{R} \leq 
	\frac{e^{\omega \tau_{\max}^i}}{R}E(t)
	\]
	for all $t \in [t_\ell, t_{\ell+1})$.
	Hence
	\begin{align}
	|g_i(t)| \leq  
	\sum_{j \in \mathcal{N}_i}
	\big| \big(x_i(t_{k_i(\ell)}^i) - x_j(t_{k_i(\ell)}^i)\big) -
	q_{ij}(t_{k_i(\ell)}^i) \big| 
	\leq
	\frac{d_i  e^{\omega \tau_{\max}^i}}{R}E(t)
	\label{eq:g_bound}
	\end{align}
	for all $t \in [t_{\ell},t_{\ell+1})$.

	From the inequalities \eqref{eq:fi_bound} and \eqref{eq:g_bound},
	we obtain
	\begin{align*}
	|f_i(t)+g_i(t)| 
	\leq \left(\delta_i +\frac{d_i  e^{\omega \tau_{\max}^i} }{R}  \right)
	E(t) 
	\leq 
	\kappa(\omega) E(t)
	\end{align*}
	for all $t \in [t_{\ell}, t_{\ell+1})$ and $i \in \mathcal{N}$.
	This and Lemma~\ref{lem:semi_norm}.a) with $F = I$ give
	\[
	\vertiii{f(t)+g(t)}_{\infty} \leq \Gamma_{\infty}  \|f(t)+g(t)\|_{\infty} \leq
	\Gamma_{\infty}   \kappa(\omega)
	E(t)
	\]
	for all $t \in [t_\ell, t_{\ell+1})$.
	Therefore,
	we have from Lemma~\ref{lem:semi_norm}.c) and d) that 
	\begin{align}
	\vertiii{x(t_{\ell}+\tau)}_{\infty} &\leq e^{-\gamma \tau} \vertiii{x(t_{\ell})}_{\infty} +
	\Gamma_{\infty} \kappa(\omega)  \int^\tau_0 e^{-\gamma (\tau - s)} 
	E(t_{\ell} + s) ds \notag \\
	&\leq 
	\left(
	e^{-\gamma \tau}+ 
	2\Gamma_{\infty} \kappa(\omega)  \int^\tau_0 e^{-\gamma (\tau - s)}
	e^{-\omega s}  ds
	\right)  \frac{E(t_\ell)}{2} \notag \\
	&  = 
	\left(
	\left(
	1 - \frac{2\Gamma_{\infty}\kappa(\omega)  }{\gamma - \omega}
	\right)
	e^{-\gamma \tau} 
	+
	\frac{2\Gamma_{\infty} \kappa(\omega) }{\gamma - \omega} e^{-\omega \tau} 
	\right) \frac{E(t_\ell)}{2} 
	\label{eq:x_tl_tau_bound}
	\end{align}
	for all $\tau \in [0,t_{\ell+1} - t_{\ell}]$.
	Since the condition \eqref{eq:omega_cond} on $\omega$ yields
	\[
	0<
	\frac{2 \Gamma_{\infty} \kappa(\omega)}{\gamma - \omega} 
	\leq 1,
	\]
	it follows that 
	\begin{equation}
	\label{eq:x_tl_tau_bound_coeff}
	\left(
	1 - \frac{2\Gamma_{\infty} \kappa(\omega) }{\gamma - \omega}
	\right)
	e^{-\gamma \tau} 
	+
	\frac{2\Gamma_{\infty} \kappa(\omega) }{\gamma - \omega} e^{-\omega \tau} 
	\leq e^{-\omega \tau}.
	\end{equation}
	Combining the inequalities \eqref{eq:x_tl_tau_bound}
	and \eqref{eq:x_tl_tau_bound_coeff}, we obtain
	\[
	\vertiii{x(t_{\ell}+\tau)}_{\infty}
	\leq 
	e^{-\omega \tau} \frac{E(t_\ell)}{2} = \frac{E(t_{\ell}+\tau)}{2} 
	\]
	for all $\tau \in [0,t_{\ell+1} - t_{\ell}]$.
	Thus $\vertiii{x(t)}_{\infty} \leq E(t)/2$
	for all $t \in [0,t_{\ell+1}]$.
\end{proof}

The condition $0< \omega \leq \gamma -2\Gamma_{\infty}\kappa(\omega)$ obtained in Lemma~\ref{lem:fundamental_lemma}
is in implicit form with respect to the decay parameter $\omega$.
We rewrite this condition in explicit form
by using 
the Lambert $W$-function. To this end,
we define
\begin{equation}
\label{eq:tilde_omega_cond}
\tilde \omega \coloneqq \min \left\{
\eta_i - \frac{W(\xi_i \tau_{\max}^i e^{\eta_i \tau_{\max}^i})}{\tau_{\max}^i} :
i \in \mathcal{N} \right\},
\end{equation}
where 
\begin{align*}
\xi_i \coloneqq \frac{2\Gamma_{\infty} d_i}{R},\quad
\eta_i \coloneqq \gamma - 2\Gamma_{\infty} \delta_i
\end{align*}
for $i \in \mathcal{N}$.
Note also that
\[
\gamma - 2 \Gamma_{\infty} \kappa(\omega) \leq
\gamma - 2 \Gamma_{\infty}
\left(
\delta_i + \frac{d_i}{R}
\right)
\]
for all $i \in \mathcal{N}$.
Therefore, if the inequality 
$0 < \gamma - 2 \Gamma_{\infty} \kappa(\omega) $ holds, then
one has
\begin{equation}
\label{eq:threshold_q_level_cond}
\delta_i  + \frac{d_i}{R} < \frac{\gamma}{2 \Gamma_{\infty}}
\end{equation}	
for all $i \in \mathcal{N}$.

\begin{lemma}
	\label{lem:omega_cond}
	Assume that the threshold 
	$\delta_i>0$ and the number $R \in \mathbb{N}$
	of quantization levels satisfy
	the inequality \eqref{eq:threshold_q_level_cond}
	for all $i \in \mathcal{N}$.
	Then $\tilde \omega$
	defined by
	\eqref{eq:tilde_omega_cond} satisfies $\tilde \omega >0$.
	Moreover, 
	the decay parameter $\omega$ satisfies 
	the condition \eqref{eq:omega_cond} 
	if and only if
	$0< \omega \leq \tilde \omega$.
\end{lemma}
\begin{proof}
	Let $i \in \mathcal{N}$.
	The inequality \eqref{eq:threshold_q_level_cond} is equivalent to
	\[
	\gamma - 2\Gamma_{\infty} \left(
	\delta_i +\frac{d_i }{R}
	\right) >0.
	\]
	Since
	\[
	\eta_i - \xi_i e^{\omega \tau_{\max}^i}=
	\gamma - 2\Gamma_{\infty} \left(
	\delta_i +\frac{d_ie^{\omega \tau_{\max}^i} }{R}
	\right),
	\]
	it follows that for all sufficiently small $\omega >0$,
	the inequality
	\begin{equation}
	\label{eq:omega_cond_i}
	\omega \leq \eta_i - \xi_i e^{\omega \tau_{\max}^i}
	\end{equation}
	holds.
	The inequality \eqref{eq:omega_cond_i}
	is equivalent to
	\[
	\xi_i \tau_{\max}^i e^{\eta_i \tau_{\max}^i}
	\leq (\eta_i - \omega)\tau_{\max}^i  e^{(\eta_i - \omega)\tau_{\max}^i}.
	\]
	Therefore,
	using the Lambert $W$-function,
	one can write the inequality \eqref{eq:omega_cond_i} as
	\[
	\omega \leq \eta_i - \frac{W(\xi_i\tau_{\max}^i  e^{\eta_i \tau_{\max}^i})}{\tau_{\max}^i}.
	\]
	Since \eqref{eq:omega_cond_i} holds for 
	all sufficiently small $\omega >0$,
	we obtain $\tilde \omega >0$.
	
	By definition,
	\[
	\gamma -2\Gamma_{\infty}\kappa(\omega) = 
	\min \{
	\eta_i - \xi_i e^{\omega \tau_{\max}^i} : i \in \mathcal{N}
	\}.
	\]
	From this, it follows that $\omega \leq \gamma - 2\Gamma_{\infty}\kappa(\omega)$
	if and only if \eqref{eq:omega_cond_i} holds for all $i \in \mathcal{N}$.
	We have shown that  \eqref{eq:omega_cond_i} holds for all $i \in \mathcal{N}$ if and only if $\omega \leq \tilde \omega$. Thus,
	the condition \eqref{eq:omega_cond} is equivalent to $0< \omega \leq \tilde \omega$.
\end{proof}

\subsection{Main result}
Before stating the main result of this section,
we summarize the assumption on the parameters of
the quantization scheme and the triggering mechanism.
\begin{assumption}
	\label{assump:paramter}
	Let upper bounds $\tau^i_{\max} >0$ be given for all 
	$i \in \mathcal{N}$.
	The following three conditions are satisfied:
	\begin{enumerate}
		\renewcommand{\labelenumi}{\alph{enumi})}
		\item
		The threshold 
		$\delta_i>0$ 
		and the number $R \in \mathbb{N}$
		of quantization levels satisfy the inequality
		\eqref{eq:threshold_q_level_cond}
		for all $i \in \mathcal{N}$.
		\item 
		For all $i \in \mathcal{N}$,
		the lower bound $\tau_{\min}^i$  satisfies
		\[
		0< \tau_{\min}^i \leq 
		\min \{ \tilde \tau_{\min}^i, \,\tau_{\max}^i\},
		\] 
		where
		$\tilde \tau_{\min}^i$ is as in \eqref{eq:tau_min_def}.
		\item
		The decay parameter $\omega$ of the quantization range $E(t)$ defined by \eqref{eq:E_def} satisfies 
		$0 < \omega  \leq \tilde \omega$, where $\tilde \omega$
		is as in
		\eqref{eq:tilde_omega_cond}.
	\end{enumerate}
\end{assumption}
\begin{theorem}
	\label{thm:main_result}
	Suppose that Assumptions 
	\ref{assump:connectivity}--\ref{assump:quantization_number} and \ref{assump:paramter} hold.
	Then the unsaturation condition
	\eqref{eq:E_bound} is satisfied 
	for all $k \in \mathbb{N}_0$ and $i \in \mathcal{N}$.
	Moreover, $\Sigma_{\mathrm{MAS}}$
	achieves consensus exponentially  with decay rate $\omega$.
\end{theorem}

\begin{proof}
	Since $0 < \omega  \leq \tilde \omega$,
	Lemma~\ref{lem:omega_cond} shows that 
	the condition \eqref{eq:omega_cond} on $\omega $ is satisfied.
	By Lemmas~\ref{lem:vivj_bound} and 
	\ref{lem:fundamental_lemma}, 
	we obtain
	\begin{equation}
	\label{eq:xij_E_bound}
	|x_i(t) - x_j(t)| \leq E(t) 
	\end{equation}
	for all $t \geq 0$ and $i,j \in \mathcal{N}$. Therefore,
	the unsaturation condition
	\eqref{eq:E_bound} is satisfied 
	for all $k \in \mathbb{N}_0$ and $i \in \mathcal{N}$.
	The inequality \eqref{eq:xij_E_bound} and 
	the definition \eqref{eq:E_def} of $E(t)$
	give
	\[
	|x_i(t) - x_j(t)| \leq  2\Gamma_{\infty} E_0 e^{-\omega t}
	\]
	for all $t \geq 0$ and $i,j \in \mathcal{N}$. Thus,
	$\Sigma_{\mathrm{MAS}}$
	achieves consensus exponentially  with decay rate $\omega$.
\end{proof}

Recall that the maximum decay parameter $\tilde \omega$
is the minimum of
\begin{equation*}
\eta_i - \frac{W(\xi_i e^{\eta_i \tau_{\max}^i})}{\tau_{\max}^i},\quad 
i \in \mathcal{N},
\end{equation*}
which is the solution of the  equation
$
\omega = \eta_i - \xi_i e^{\omega \tau_{\max}^i};
$
see the proof of Lemma~\ref{lem:omega_cond}.
Moreover, $\xi_i$  becomes smaller as $d_i/R$ decreases, and
$\eta_i$ becomes larger as $\delta_i$ decreases.
Therefore,  $\tilde \omega$ becomes larger as
$d_i$, $\delta_i$, and $\tau_{\max}^i$ decreases and as 
$R$ increases. This also means that 
if agent~$i$ has a large $d_i$, i.e., many neighbors, then
we need to use small $\delta_i$ and $\tau_{\max}^i$ in order to
achieve fast consensus of the multi-agent system.

\begin{remark}
	\normalfont
	The condition on the lower bound $\tau_{\min}^i$  in Assumption~\ref{assump:paramter}.b) is not used when
	each agent computes the next sampling time; see Section~\ref{sec:self_triggered_computation}
	for details. 
	Therefore, Theorem~\ref{thm:main_result} essentially shows that 
	asymptotic consensus is achieved
	if 
	\eqref{eq:threshold_q_level_cond} holds for each $i \in \mathcal{N}$ and 
	if $0 < \omega \leq \tilde \omega$
	for given upper bounds $\tau^1_{\max},\dots,\tau^N_{\max}$
	of inter-event times.
\end{remark}
\begin{remark}
	\normalfont
	\label{rem:global_data}
	To check the conditions obtained in Theorem~\ref{thm:main_result}, 
	the global network parameters,  $\lambda_2(L)$ and
	$\Gamma_{\infty}$, are needed. In addition,
	the quantization range $E(t)$ is common to all agents as 
	the scaling parameter of finite-level dynamic quantizers studied,
	e.g., in the previous works~\cite{Li2011, You2011}.
	These are drawbacks of the proposed method.
\end{remark}

\begin{remark}
	\normalfont
	Although
	the proposed method is inspired by the self-triggered consensus algorithm presented
	in~\cite{Fan2015}, 
	the approach to consensus analysis differs.
	In \cite{Fan2015}, a Lyapunov function and LaSalle's invariance principle have been
	employed. In contrast,
	we develop a trajectory-based approach, where 
	the semi-contractivity property of $e^{-Lt}$ plays a key role.
	Moreover, we discuss the convergence speed of consensus, by using
	the global parameters mentioned in Remark~\ref{rem:global_data} above.
	The utilization of the global  parameters also enables us to investigate
	the minimum inter-event time in a way different from that of \cite{Fan2015}.
\end{remark}

\subsection{Bounds of $\Gamma_{\infty}$}
We use the constant $\Gamma_{\infty}$ in the definition \eqref{eq:E_def} of $E(t)$ and the conditions for consensus given in Assumption~\ref{assump:paramter}.
To apply the proposed method, 
we have to compute
$\Gamma_{\infty}$ numerically by \eqref{eq:Gamma_inf_def} or 
replace $\Gamma_{\infty}$ with an available upper bound of $\Gamma_{\infty}$.
In the next proposition, we provide bounds of $\Gamma_{\infty}$ by using the network size.
The proof can be found in Appendix A.

\begin{proposition}
	\label{prop:Gamma_inf_bound}
	Let $N \in \mathbb{N}$ satisfy $N \geq 2$ and
	let $G$ be a connected undirected graph with $N$ vertices.
	Define $L \coloneqq L(G)$.
	Then the following statements hold for $\Gamma_{\infty}(\gamma)$
	defined as in \eqref{eq:Gamma_inf_def}: 
	\begin{enumerate}
		\renewcommand{\labelenumi}{\alph{enumi})}
		\item 
		For all $0 < \gamma \leq \lambda_2(L)$, 
		\begin{equation*}
		2-\frac{2}{N} \leq \Gamma_{\infty}(\gamma) \leq N - 1.
		\end{equation*}
		\item If $G$ is a complete graph, then
		\[
		\Gamma_{\infty}(\gamma) = 2-\frac{2}{N}
		\]
		for all $0 < \gamma \leq \lambda_2(L) = N$.
	\end{enumerate}
\end{proposition}

We conclude this section by using
Proposition~\ref{prop:Gamma_inf_bound}.b) 
to examine the relationship between
the network size of complete graphs and the design parameters 
for quantization and self-triggered sampling.
For real-valued functions $\Phi,\Psi$ on $\mathbb{N}$, we write
\[
\Phi(N) = \Theta\big(\Psi (N) \big)\quad \text{as $N \to \infty$}
\]
if there are $C_1,C_2 >0$ and $N_0 \in \mathbb{N}$ such that 
for all $N \geq N_0$,
\[
C_1 \Psi(N) \leq \Phi(N) \leq C_2\Psi(N).
\] 
\begin{example}
	\normalfont
	Let $G$ be a complete graph with $N$ vertices.
	
	\textit{Sensing accuracy:}
	By Proposition~\ref{prop:Gamma_inf_bound}.b), one 
	can set 
	\[
	\gamma = N
	\text{~~and~~}
	\Gamma_{\infty} = 2-\frac{2}{N}.
	\]
	We see from the condition \eqref{eq:threshold_q_level_cond}  that if
	the number $R$ of the quantization levels satisfies
	\begin{equation}
	\label{eq:N_cond_comp}
	R >  \frac{2d_i \Gamma_{\infty}}{\gamma} = \frac{4(N-1)^2}{N^2},
	\end{equation}
	then the quantized self-triggered multi-agent system
	achieves consensus exponentially for some threshold $\delta_i$.
	Hence,  
	the required sensing accuracy  for asymptotic consensus is $\Theta(1)$ as 
	$N \to \infty$.
	
	\textit{Number of indices for data transmission:}
	Recall that the agents send the sum of 
	relative state measurements to all neighbors for the computation of
	sampling times.
	The number of indices used for this communication 
	is 
	\[
	2\tilde d R_0+1,
	\] 
	where $R_0\in \mathbb{N}_0$  and
	$\tilde d\in \mathbb{N}$ 
	satisfy 
	$R = 2R_0 + 1$ and 
	\[
	d_i = N-1 \leq \tilde d
	\]
	for all $i \in \mathcal{N}$, respectively.
	Hence, the required number of indices for asymptotic consensus
	is $\Theta(N)$ as 
	$N \to \infty$.
	
	\textit{Threshold for sampling:}
	We see from the condition \eqref{eq:threshold_q_level_cond} 
	that the threshold $\delta_i$
	of the triggering mechanism \eqref{eq:STM} 
	of agent~$i$ has to satisfy
	\[
	\delta_i < 
	\frac{\gamma}{2\Gamma_{\infty}} -
	\frac{d_i}{R} = 
	\frac{N^2}{2(N-1)} - \frac{N-1}{R}.
	\]
	Combining this inequality with \eqref{eq:N_cond_comp},
	we have that the required threshold for asymptotic consensus
	is $\Theta(N)$ 
	as 
	$N \to \infty$.
\end{example}

\section{Computation of Sampling Times}
\label{sec:self_triggered_computation}
In this section,
we describe how the agents  compute sampling times 
in a self-triggered fashion.
We discuss an initial candidate of the next sampling time and then
the first update of the candidate, followed by the $p$-th update.
Finally, we present a  joint algorithm for 
quantization and self-triggering sampling.

Let 
$i \in \mathcal{N}$ and 
$k \in \mathbb{N}_0$. Define $\tilde \tau_{\min}^i$ 
by
\eqref{eq:tau_min_def_usemin}.
By Proposition~\ref{prop:basic_property}.d) and f), there exists $\ell_0 \in \mathbb{N}_0$
such that $t_k^i = t_{\ell_0}$.

\subsection{Initial candidate of the next sampling time}
\label{sec:initial_candidate}
First,
agent~$i$ updates $q_i$ at time $t = t_{\ell_0}$.
If the neighbor~$j$ also updates 
$q_j$ at time $t = t_{\ell_0}$, then
agent~$i$ receives $q_j$.
Next, agent~$i$ computes
a candidate of the inter-event time,
\[
\tau_{k,0}^i \coloneqq \min\{ \tilde \tau_{k,0}^i, \,\tau_{\max}^i \},
\]
where
\begin{align*}
\tilde \tau_{k,0}^i \coloneqq 
\inf\Bigg\{
\tau >0:
\Bigg| \tau d_i q_i(t_k^i) - \tau \sum_{j \in \mathcal{N}_i}
q_j(t^j_{k_j(\ell_0)}) \Bigg|   \geq  \delta_i e^{-\omega \tau} E(t_{\ell_0})
\Bigg\} .
\end{align*}
By \eqref{eq:qi_bound} and \eqref{eq:qj_bound}, $\tilde \tau_{k,0}^i \geq \tilde \tau_{\min}^i$.
Agent~$i$ takes $t_k^i + \tau_{k,0}^i$
as an initial candidate of the next sampling time.
If 
agent~$i$ does not receive an updated $q_j$ from any neighbors~$j$
on the interval $(t_k^i,t_k^i + \tau_{k,0}^i)$,
then $t_k^i +\tau_{k,0}^i $ is the next sampling time,
that is, agent~$i$ updates $q_i$ at $t = t_k^i +\tau_{k,0}^i $.

Using the Lambert $W$-function,
one can write $\tau_{k,0}^i$ more explicitly.
To see this, we first note that the solution $\tau = \tau^*$
of the equation 
\[
a\tau + c = b e^{-\omega \tau},\quad a,b>0,\, c \in \mathbb{R}
\] 
is written as 
\[
\tau^* =  \dfrac{1}{\omega}W \left(
\dfrac{\omega b}{a } e^{
	\omega c /a
}
\right) - \dfrac{c }{a}.
\]
Define  the function $\phi_0$ by
\begin{equation*}
\phi_0(a,b,c) \coloneqq 
\begin{cases}
\dfrac{1}{\omega}W \left(
\dfrac{\omega b}{|a| } e^{
	\omega c /a
}
\right) - \dfrac{c }{a} & \text{if $a \not=0$} \vspace{6pt} \\
\dfrac{1}{\omega} \log
\dfrac{b}{|c|}
& \text{if $a =0$ and $c \not= 0$}\vspace{4pt} \\
\infty & \text{if $a =0$ and $c = 0$}
\end{cases}
\end{equation*}
for $a,c \in \mathbb{R}$ and $b >0$.
We also set
\begin{align}
a_{k,0}^i &\coloneqq 
d_i q_i(t_k^i) - \sum_{j \in \mathcal{N}_i}
q_j(t^j_{k_j(\ell_0)})  \notag \\
b_{k,0}^i &\coloneqq \delta_i E(t_{\ell_0}) 
\label{eq:abck0}\\
c_{k,0}^i &\coloneqq 0. \notag 
\end{align}
Since
\begin{align*}
\tilde \tau_{k,0}^i =
\inf\big\{
&\tau >0: 
|a_{k,0}^i| \tau + c_{k,0}^i
\geq   b_{k,0}^i e^{-\omega \tau} 
\big\},
\end{align*}
we have 
$
\tilde \tau_{k,0}^i = 
\phi_0(a_{k,0}^i,b_{k,0}^i,c_{k,0}^i) 
$.
Hence
\begin{equation}
\label{eq:til_tauk0}
\tau_{k,0}^i = \min \{\phi_0(a_{k,0}^i,b_{k,0}^i,c_{k,0}^i) ,\,\tau_{\max}^i   \}.
\end{equation}

\subsection{First update}
\label{sec:sec_up}
If agent~$i$ receives an updated $q_j$ from some neighbor~$j$
by $t = t_k^i + \tau_{k,0}^i $,
then agent~$i$ must recalculate a candidate of the next sampling
time as in the self-triggered method proposed in \cite{Fan2015}.
We will now consider this scenario, i.e.,  the case
\[
\{
\ell \in \mathbb{N}:
t_k^i < t_{\ell} < t_k^i + \tau_{k,0}^i \text{~~and~~} \mathcal{I}(\ell) \cap 
\mathcal{N}_i  \not= \emptyset \} \not= \emptyset,
\]
where $\mathcal{I}(\ell)$ is defined as in Section~\ref{sec:sampling}.
Let $t_{\ell_1} \in (t_k^i, t_k^i + \tau_{k,0}^i)$ 
be the first instant at which 
agent~$i$ receives updated data
after $t=t_k^i$.
Since $t_k^i = t_{\ell_0}$, one can write $\ell_1$ as
\[
\ell_1 =
\min \{ \ell \in \mathbb{N}: \ell> \ell_0\text{~~and~~}\mathcal{I}(\ell) \cap 
\mathcal{N}_i  \not= \emptyset \}.
\]
Note that 
agent~$i$ may receive updated data from several neighbors at time
$t = t_{\ell_1}$.

By using the new data,
agent~$i$ computes the following inter-event time at time
$t = t_{\ell_1}$:
\[
\tau_{k,1}^i \coloneqq \min \{\tilde 
\tau_{k,1}^i + (t_{\ell_1}-t_{\ell_0}) , \, \tau_{\max}^i   \},
\]
where
\begin{align*}
\tilde \tau_{k,1}^i \coloneqq 
\inf\Bigg\{
\tau >0: \Bigg|
\tau d_i q_i(t_k^i) - \tau \sum_{j \in \mathcal{N}_i}q_j(t^j_{k_j(\ell_1)}) +(t_{\ell_1} - t_{\ell_0}) a_{k,0}^i 
\Bigg| \geq  \delta_i e^{-\omega \tau} E(t_{\ell_1})
\Bigg\} .
\end{align*}
Then $t_k^i +  \tau_{k,1}^i $ is 
a new candidate of the next sampling time.
By \eqref{eq:qi_bound} and \eqref{eq:qj_bound}, we obtain
\[
\tilde \tau_{k,1}^i+ (t_{\ell_1}-t_{\ell_0}) \geq \tilde \tau_{\min}^i.
\]
As in the initial case, if
agent~$i$ does not receive an updated $q_j$ from any neighbors~$j$
on the interval $(t_{\ell_1},t_k^i + \tau_{k,1}^i)$,
then $t_k^i + \tau_{k,1}^i$ is the next sampling time.
Otherwise, agent~$i$ computes the next sampling time again
in the same way.

One can rewrite $\tilde \tau_{k,1}^i$ by using the Lambert $W$-function.  
To see this, we define
\begin{align*}
a_{k,1}^i &\coloneqq 
d_i q_i(t_k^i) - \sum_{j \in \mathcal{N}_i}
q_j(t^j_{k_j(\ell_1)})  \\
b_{k,1}^i &\coloneqq \delta_i E(t_{\ell_1}) \\
c_{k,1}^i &\coloneqq (t_{\ell_1} - t_{\ell_0}) a_{k,0}^i.
\end{align*}
Then
\begin{align*}
\tilde \tau_{k,1}^i =
\inf\big\{
&\tau >0: \big|
a_{k,1}^i \tau + c_{k,1}^i
\big| \geq  b_{k,1}^i e^{-\omega \tau} 
\big\} .
\end{align*}
From the definition of $\tilde \tau_{k,0}^i$ and $t_{\ell_1} - t_{\ell_0} <
\tilde \tau_{k,0}^i$, we obtain
\[
|c_{k,1}^i| <  b_{k,1}^i.
\]
If the product $a_{k,1}^i  c_{k,1}^i$ satisfies  $a_{k,1}^i  c_{k,1}^i 
\geq 0$, then the condition
\begin{equation}
\label{eq:abc_cond}
\big|
a_{k,1}^i \tau + c_{k,1}^i
\big| \geq b_{k,1}^i e^{-\omega \tau} 
\end{equation}
can be written as
\[
|a_{k,1}^i| \tau + |c_{k,1}^i| \geq b_{k,1}^i e^{-\omega \tau},
\]
and hence $\tilde \tau_{k,1}^i =
\phi_0 (a_{k,1}^i,b_{k,1}^i,c_{k,1}^i )$.

Next we consider the case $a_{k,1}^i  c_{k,1}^i < 0$.
In this case,  the condition \eqref{eq:abc_cond}
is equivalent to
\begin{equation*}
\begin{cases}
-|a_{k,1}^i| \tau + |c_{k,1}^i| \geq  b_{k,1}^i e^{-\omega \tau} & \text{for $\displaystyle 
	\tau \leq  -c_{k,1}^i/a_{k,1}^i$} \vspace{3pt}\\
|a_{k,1}^i| \tau - |c_{k,1}^i| \geq b_{k,1}^i e^{-\omega \tau} &  \text{for $\tau > -c_{k,1}^i/a_{k,1}^i$}.
\end{cases}
\end{equation*}
For the latter inequality,
we have that
\begin{align*}
\inf\{ \tau > -c_{k,1}^i/a_{k,1}^i :
|a_{k,1}^i| \tau - |c_{k,1}^i| \geq b_{k,1}^i e^{-\omega \tau} 
\}
&=\inf\{ \tau > 0 :
|a_{k,1}^i| \tau - |c_{k,1}^i| \geq b_{k,1}^i e^{-\omega \tau} 
\}\\
& =
\phi_0 (a_{k,1}^i,b_{k,1}^i,c_{k,1}^i ).
\end{align*}
It may also occur that 
\[
\{
0< \tau \leq  -c_{k,1}^i/a_{k,1}^i:
-|a_{k,1}^i| \tau + |c_{k,1}^i| \geq  b_{k,1}^i e^{-\omega \tau}\}
\not= \emptyset.
\]
To see this, we first observe that
\begin{align}
\{
0< \tau \leq  -c_{k,1}^i/a_{k,1}^i:
-|a_{k,1}^i| \tau + |c_{k,1}^i| \geq b_{k,1}^i e^{-\omega \tau}\} =
\{
\tau >0:
-|a_{k,1}^i| \tau + |c_{k,1}^i| \geq b_{k,1}^i e^{-\omega \tau}\}.
\label{eq:ac_negative_set}
\end{align}
Let $W_{-1}$ be the secondary branch of the Lambert $W$-function, i.e.,
$W_{-1}(y)$ is the solution $x \leq -1$ of the equation 
$xe^{x} = y$ for $y \in [-e^{-1},0)$.
We obtain the infimum of the set in
\eqref{eq:ac_negative_set} from the following proposition, whose proof
is given in Appendix B.
\begin{proposition}
	\label{prop:W1_case}
	Let $0 < c < b$. Then
	\begin{align*}
	\inf \{
	\tau >0:
	-\tau + c \geq  b e^{-\tau}\}  =
	\begin{cases}
	W_{-1}(-be^{-c}) + c & \text{if $1 < b \leq e^{-1+c}$} \\
	\infty & \text{otherwise}.
	\end{cases}
	\end{align*}
\end{proposition}
To apply Proposition~\ref{prop:W1_case}, note that
\[
-|a_{k,1}^i| \tau + |c_{k,1}^i| \geq  b_{k,1}^i e^{-\omega \tau} 
\]
if and only if
\[
- \hat \tau - \frac{\omega c_{k,1}^i}{a_{k,1}^i} \geq 
\frac{ \omega b_{k,1}^i}{|a_{k,1}^i|}  
e^{-\hat \tau},
\quad \text{where $\hat \tau \coloneqq \omega \tau$}.
\]
Define the function $\phi$ by
\[
\phi(a,b,c) \coloneqq
\begin{cases}
\dfrac{1}{\omega}W_{-1} \left(
-\dfrac{\omega b}{|a| } e^{
	\omega c /a
}
\right) - \dfrac{c }{a} & \text{if $(a,b,c) \in \Upsilon_{\omega}$} \vspace{4pt}\\
\phi_0(a,b,c) & \text{if $(a,b,c) \not\in \Upsilon_{\omega}$}
\end{cases}
\]
for $a,c, \in \mathbb{R}$ and $b>0$,
where the set $\Upsilon_{\omega}$ is given by
\[
\Upsilon_{\omega} \coloneqq
\left\{
(a,b,c) : ac <0 \text{~~and~~} 
1 < \frac{\omega b}{|a|}  \leq e^{-1-\omega c/a}
\right\}.
\]
From Proposition~\ref{prop:W1_case},
we conclude that 
\[
\tilde \tau_{k,1}^i = 
\phi (a_{k,1}^i,b_{k,1}^i,c_{k,1}^i )
\]
in both cases  $a_{k,1}^i  c_{k,1}^i 
\geq 0$ and  $a_{k,1}^i  c_{k,1}^i 
< 0$.

\subsection{$p$-th update}
Let $p \in \mathbb{N}$ and let  
\[
t_{\ell_0} < t_{\ell_1} < \dots < t_{\ell_p} < t_{\ell_0} +\tau_{\max}^i.
\]
We consider the case where
agent~$i$ receives new data from its
neighbors
at times $t = t_{\ell_1},\dots,t_{\ell_p}$ before 
the next candidate sampling times.

At time $t = t_{\ell_p}$, agent~$i$ computes
\begin{align}
\label{eq:til_taukp}
\tau_{k,p}^i \coloneqq
\min \{\phi (a_{k,p}^i,b_{k,p}^i,c_{k,p}^i )+ (t_{\ell_p}-t_{\ell_0}) ,\, 
\tau_{\max}^i   \},
\end{align}
where 
\begin{align}
a_{k,p}^i &\coloneqq  d_i q_i(t_k^i) - \sum_{j \in \mathcal{N}_i}
q_j(t^j_{k_j(\ell_p)}) \notag \\
b_{k,p}^i &\coloneqq \delta_i E(t_{\ell_p}) \label{eq:abckp}\\
c_{k,p}^i &\coloneqq  c_{k,p-1}^i + 
(t_{\ell_{p}} - t_{\ell_{p-1}}) a_{k,p-1}^i, \notag 
\end{align}
and takes $t_k^i + \tau_{k,p}^i$ as a new 
candidate of the next sampling time.
We have 
\begin{equation}
\label{eq:phi_lower_bound}
\phi (a_{k,p}^i,b_{k,p}^i,c_{k,p}^i )+ (t_{\ell_p}-t_{\ell_0}) \geq 
\tilde \tau_{\min}^i
\end{equation}
as  in the first update explained above.
Since $\tau_{k,p}^i$
satisfies 
\[
\tilde \tau_{\min}^i \leq \tau_{k,p}^i \leq \tau_{\max}^i,
\]
only a finite number of data transmissions 
from neighbors occur until the next sampling time.

The next theorem shows that when the neighbors 
do not update the measurements on the interval
$(t_{\ell_p}, t_k^i + \tau_{k,p}^i)$, 
the candidate $t_k^i + \tau_{k,p}^i$  of the next sampling times
constructed as above coincides with the next sampling time $t_{k+1}^i$ computed from the triggering mechanism \eqref{eq:STM}.
\begin{theorem}
	\label{thm:sampling_time_coincidence}
	Let $i \in \mathcal{N}$ and $k,p \in \mathbb{N}_0$. 
	Let $t_{\ell_0},\dots,t_{\ell_p}$ and $\tau_{k,p}^i$ be as above, and assume that 
	agent~$i$ does not receive any measurements
	from its neighbors on the interval
	$(t_{\ell_p}, t_k^i + \tau_{k,p}^i)$. 
	Then
	\begin{equation}
	\label{eq:next_sampling_time}
	t_k^i + \tau_{k,p}^i = t_{k+1}^i,
	\end{equation}
	where $t_{k+1}^i$ is defined by \eqref{eq:STM} with
	$
	0< \tau_{\min}^i \leq 
	\min \{ \tilde \tau_{\min}^i, \,\tau_{\max}^i\}.
	$
\end{theorem}
\begin{proof}
	By the definition of $t_{\ell_p}$, we obtain
	\[
	|f_k^i(\tau)| <  \delta_i E( t_{\ell_0} + \tau)
	\]
	for all $\tau \in [0,t_{\ell_p}-t_{\ell_0}]$.
	Moreover, the arguments given in Sections~\ref{sec:initial_candidate}
	and \ref{sec:sec_up} show that
	\begin{align*}
	\phi (a_{k,p}^i,b_{k,p}^i,c_{k,p}^i )   = \inf\{
	\tau > 0:
	|f_k^i(t_{\ell_p}-t_{\ell_0} + \tau )| \geq  \delta_i E(t_{\ell_p}+\tau)
	\}.
	\end{align*}
	From these facts, it follows that
	\begin{align*}
	\phi (a_{k,p}^i,b_{k,p}^i,c_{k,p}^i ) + (t_{\ell_p}-t_{\ell_0})  
	&= \inf\{
	\tau > t_{\ell_p}-t_{\ell_0}:
	|f_k^i(\tau )| \geq  \delta_i E(t_{\ell_0}+\tau)
	\} \\
	& =
	\inf\{
	\tau > 0:
	|f_k^i(\tau )| \geq  \delta_i E(t_{\ell_0}+\tau)
	\}.
	\end{align*}
	Combining this
	with the inequality \eqref{eq:phi_lower_bound},
	we obtain
	\begin{align*}
	\phi (a_{k,p}^i,b_{k,p}^i,c_{k,p}^i )+ (t_{\ell_p}-t_{\ell_0})  =
	\inf\{
	\tau \geq \tau^i_{\min} :
	|f_k^i(\tau)| \geq \delta_i E(t_k^i+\tau)
	\} = \tau_k^i
	\end{align*}
	for all $0< \tau_{\min}^i \leq \min \{ \tilde \tau_{\min}^i, \,\tau_{\max}^i\}$,
	where $\tau_k^i$ is defined as in \eqref{eq:STM}.
	Thus, we obtain the desired result \eqref{eq:next_sampling_time}.
\end{proof}

\subsection{Algorithm for quantization
	and self-triggered sampling}
We are now ready to present a joint algorithm for finite-level dynamic 
quantization and self-triggered sampling. 
Under this algorithm,  
the unsaturation condition
\eqref{eq:E_bound} is satisfied 
for all $k \in \mathbb{N}_0$ and $i \in \mathcal{N}$,
and the multi-agent system achieves
consensus  exponentially  with decay rate $\omega$;
see Theorems~\ref{thm:main_result} and \ref{thm:sampling_time_coincidence}. Moreover, the inter-event times 
$t_{k+1}^i - t_k^i$
are bounded from below 
by the constant $\tilde \tau_{\min}^i > 0$  for all 
$k \in \mathbb{N}_0$ and $i \in \mathcal{N}$.

\begin{algorithm}[Action of agent~$i$ on the sampling interval $t_k^i \leq t <  
	t_{k+1}^i$] \normalfont \mbox{}\vspace{2pt} \\
	\noindent
	\textbf{Step 0.}
	Choose the threshold 
	$\delta_i>0$ 
	and the number $R=2R_0+1$, $R_0 \in \mathbb{N}_0$,
	of quantization levels such that the inequality
	\eqref{eq:threshold_q_level_cond} holds
	for all $i \in \mathcal{N}$.
	Choose
	the upper bounds $\tau^1_{\max},\dots,\tau^N_{\max}>0$
	of inter-event times and
	the decay parameter $\omega$ of the quantization 
	range $E(t)$ such that $0 < \omega \leq
	\tilde \omega$, where $\tilde \omega$ is defined as in
	\eqref{eq:tilde_omega_cond}.
	
	\vspace{2pt}
	\noindent
	\textbf{Step 1.}
	At time $t = t_k^i \eqqcolon t_{\ell_0}$, agent~$i$
	performs the following actions~i)--v).
	\begin{enumerate}
		\renewcommand{\labelenumi}{\roman{enumi})}
		\item 
		Measure the quantized relative state $q_{ij}(t_k^i)$
		for all $j \in \mathcal{N}_i$ and deactivate the sensor.
		\item
		Encode the sum $q_i(t_k^i)$ of the quantized measurements 
		to an index in a finite set with cardinality $2\tilde d R_0+1$
		and transmit the index to each neighbor~$j \in \mathcal{N}_i$.
		\item If an index is received
		from a neighbor at time $t= t_{\ell_0}$, then decode the index and update the sum of the relative state measurements of the neighbor. 
		\item 
		Compute
		$\tau_{k,0}^i$ by \eqref{eq:til_tauk0}, where 
		$a_{k,0}^i$, $b_{k,0}^i$, and $c_{k,0}^i$ 
		are defined as in \eqref{eq:abck0}.
		\item
		Set $p = 0$.
	\end{enumerate}
	
	\vspace{2pt}
	\noindent
	\textbf{Step 2.}
	Agent~$i$ plans to activate the sensor at time $t = t_k^i +  \tau_{k,p}^i$.
	
	\vspace{3pt}
	\noindent
	\textbf{Step 3-a.}
	If agent~$i$ receives an index from some neighbor 
	on the interval $(t_{\ell_p}, t_k^i + \tau_{k,p}^i)$, then
	agent~$i$ performs the following actions~i)--iii). Then go back to 
	{\bf Step 2}.
	\begin{enumerate}
		\renewcommand{\labelenumi}{\roman{enumi})}
		\item
		Set $p$ to $p+1$ and
		store the time $t_{\ell_p}$ at which the index is received. 
		\item 
		Decode the index and update
		the sum of the relative state measurements of the neighbor. 
		If several indices are received at time $t = t_{\ell_p}$,
		then this action is applied to all indices.
		\item 
		Compute $\tau_{k,p}^i$ by \eqref{eq:til_taukp},
		where $a_{k,p}^i$, $b_{k,p}^i$, and $c_{k,p}^i$ 
		are defined as in \eqref{eq:abckp}.
	\end{enumerate}
	
	\vspace{3pt}
	\noindent
	\textbf{Step 3-b.}
	If agent~$i$ does not receive any indices
	on the interval $(t_{\ell_p}, t_k^i + \tau_{k,p}^i)$, then
	agent~$i$ sets $t_{k+1}^i \coloneqq t_k^i + \tau_{k,p}^i$.
	
	\vspace{3pt}
	\noindent
	\textbf{Step 4.}
	Agent~$i$ sets $k$ to $k+1$. Then go back to {\bf Step 1}.
\end{algorithm}

\begin{remark}
	\normalfont
	The proposed method takes advantage of the simplicity of 
	the first-order dynamics in the following way.
	Assume that the dynamics of agent~$i$ is given by
	\[
	\dot x_i(t) = Ax_i(t) +Bu_i(t),
	\]
	where $A \in \mathbb{R}^{n\times n}$ and 
	$B \in \mathbb{R}^{n\times m}$. Then 
	the error $x_i(t_k^i+\tau) - x_i(t_k^i)$ due to sampling is 
	written as
	\begin{align*}
	x_i(t_k^i+\tau) - x_i(t_k^i)  = (e^{A\tau} - I) x_i(t_k^i) + \int^{\tau}_{0}
	e^{A(\tau-s)}B u_i(t_k^i+s)  ds
	\end{align*}
	for $\tau \geq 0$.
	Since $e^{A\tau} - I \not =0$ in general,
	the absolute state $x_i(t_k^i)$ is  required to describe the error $x_i(t_k^i+\tau) - x_i(t_k^i)$.
	However,  one has $e^{A\tau} - I =0$ in the first-order case $A = 0$,
	and hence 
	the absolute state $x_i(t_k^i)$ needs not be measured in the proposed algorithm.
	Moreover, since the input $u_i$ is constant on the sampling interval,
	the integral term is a  linear function with respect to $\tau$ in the first-order case $A = 0$. This enables us to use the Lambert $W$-function for the computation of sampling times.
\end{remark}

\section{Numerical simulation}
\label{sec:numerical_sim}
In this section, we consider the connected network 
shown in Figure~\ref{fig:network_example}, where the number $N$ of agents is $N=6$.
\begin{figure}[!b]
	\centering
	\includegraphics[width = 2cm]{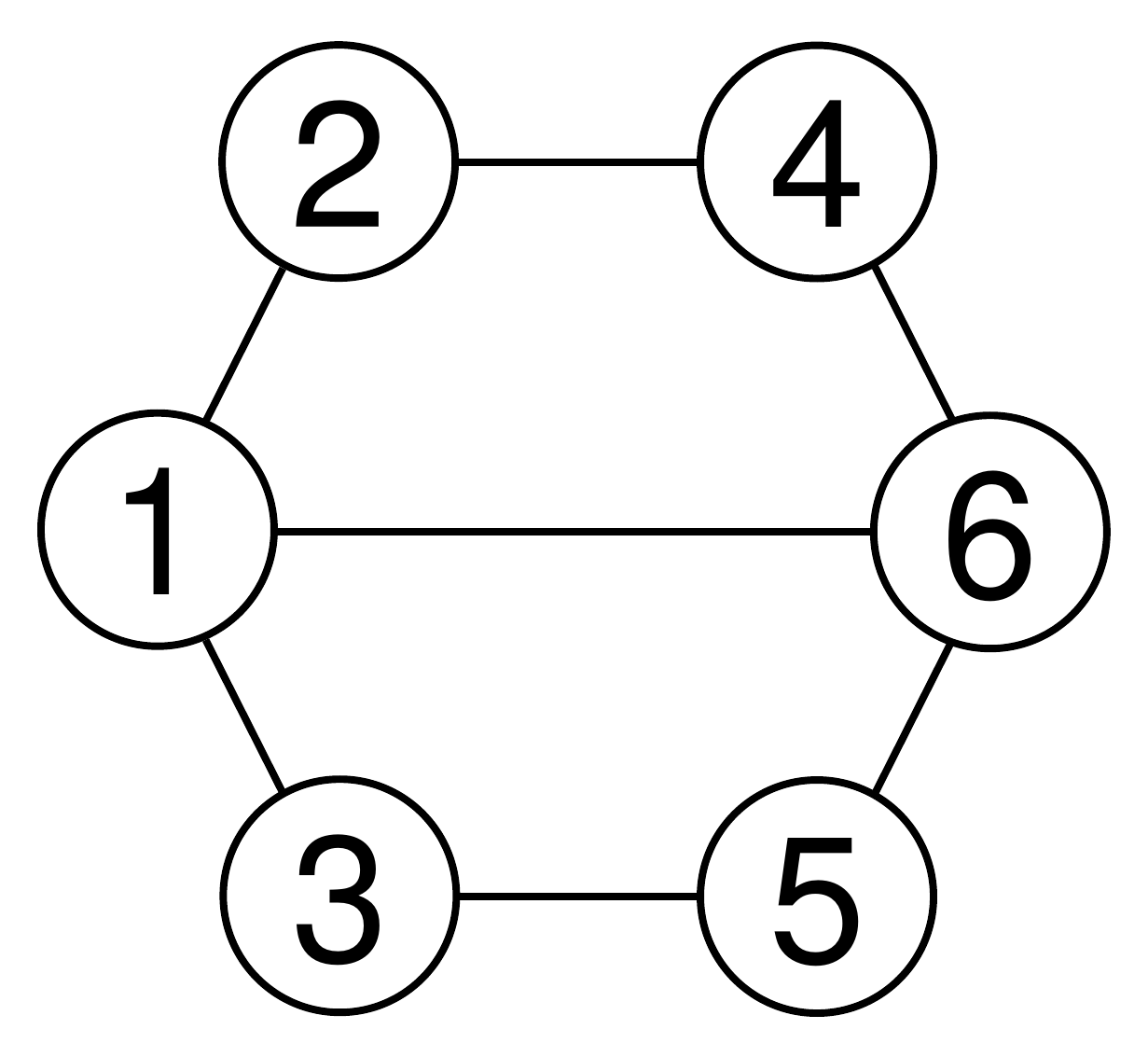}
	\caption{Network topology.}
	\label{fig:network_example}
\end{figure}
For each $i \in \mathcal{N} = \{1,2,\dots,6 \}$, 
the initial state $x_{i0}$ is given by
$x_{i0} = \sin(i)$. 
Since
\[
\max_{i\in \mathcal{N}}
\left|
x_{i0}  - \frac{1}{N} \sum_{j\in \mathcal{N}} x_{j0} 
\right| \leq 0.95,
\]
a bound $E_0$ in Assumption~\ref{assump:initial_bound}
is chosen as
$E_0 = 1$.
We set 
\[
\gamma = \lambda_2(L) = 1
\] 
and then 
numerically compute $\Gamma_{\infty} = 5/3$, where
$\Gamma_{\infty}$ is defined by \eqref{eq:Gamma_inf_def}.

The threshold $\delta_i$ and 
the upper bound $\tau_{\max}^i$ of inter-event times 
for the triggering mechanism \eqref{eq:STM}
are given by 
\begin{align*}
\delta_i = 
\begin{cases}
0.04  & \text{if $i=1,6$} \\
0.09 & \text{otherwise},
\end{cases}\qquad 
\tau_{\max}^i = 
\begin{cases}
1  & \text{if $i=1,6$} \\
1.5 & \text{otherwise},
\end{cases}
\end{align*}
respectively. The reason why agents $1$ and $6$ have
smaller thresholds and upper bounds of inter-event times 
is that 
these agents have more neighbors than others. 
For these thresholds, the minimum odd number $R$
satisfying the condition
\eqref{eq:threshold_q_level_cond} for all $i \in \mathcal{N}$
is $13$.
By Theorem~\ref{thm:main_result}, if the number $R$ 
of quantization levels is odd and satisfies $R \geq 13$, then
the multi-agent system achieves consensus exponentially
for a suitable decay parameter $\omega$ of the quantization range $E(t)$.
We use $R = 19$ for the simulation below.
Then $R_0 \in \mathbb{N}_0$ with $R = 2R_0 + 1$ is given by $R_0 = 9$.
When
each agent knows 
\[
\tilde d = 3
\]
as a bound of 
the number of neighbors, as stated 
in Assumption~\ref{assump:ni_bound},
the number of quantization levels for the transmission of
the sum 
of the relative states
is 
\[
2 \tilde d R_0 + 1= 55,
\] 
which can be represented by $6$ bits.
Under this setting of the parameters $\gamma$, $\delta_i$, $\tau_{\max}^i$,
and $R$, 
the maximum decay parameter $\tilde \omega$, 
which is defined as in \eqref{eq:tilde_omega_cond},
is given by 
\[
\tilde \omega = 0.2145.
\] 
In the simulation, 
we set $\omega = \tilde \omega$.

Using the Lambert $W$-function, we can compute
a lower bound $\tilde \tau_{\min}^i$
of inter-event times by \eqref{eq:tau_min_def}:
\begin{align*}
\tilde \tau_{\min}^i = 
\begin{cases}
2.192 \times 10^{-3}  & \text{if $i=1,6$} \\
8.574 \times 10^{-3} & \text{otherwise}.
\end{cases}
\end{align*}
Note, however, that these lower bounds are not used
for the real-time computation of inter-event times, because
all candidates of the inter-event times 
computed by the agents are greater than or equal to
these lower bounds as shown in Section~\ref{sec:self_triggered_computation}.

The state trajectory and the corresponding sampling times of each
agent are shown in Figures~\ref{fig:state_trajectory} and \ref{fig:sampling_time}, respectively, where 
the simulation time is $16$ and the time step is $10^{-4}$.
From Figure~\ref{fig:state_trajectory}, we see that the deviation of each state from the average state
converges to zero.
Figure~\ref{fig:sampling_time} shows that 
sampling occurs frequently on the interval $[0,1]$ but less
frequently on the interval $[1,16]$.
Agent~$3$ measures relative states more frequently 
on the interval $[4,7]$ than on other intervals.
This is because the state of agent~$3$ oscillates due to
coarse quantization. Such oscillations can be observed also for other agents, e.g., 
agent~$1$ on the interval $[3,4]$.
Moreover, we find in Figure~\ref{fig:state_trajectory} that the states of
agents~$2$ and $5$ do not change on the intervals $[2,4]$ and $[2,7]$, respectively. This is also caused by coarse quantization. In fact,
the quantized values of their relative state measurements
are zero on these intervals.
However, 
the proposed algorithm ensures that
the quantization errors exponentially converge to zero, and hence
the multi-agent system achieves asymptotic consensus.

\begin{figure}[!t]
	\centering
	\includegraphics[width = 8.5cm]{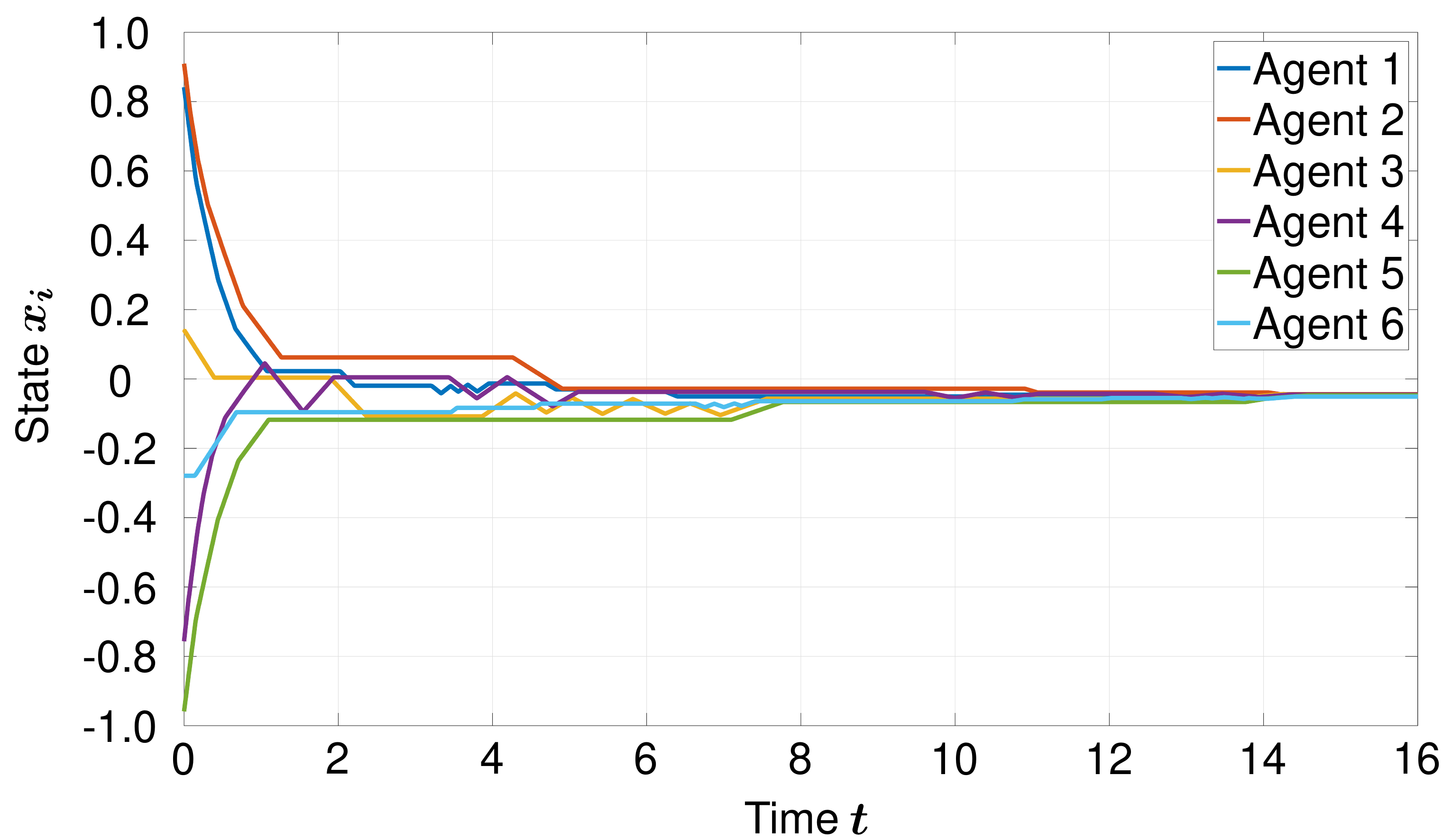}
	\caption{State trajectories.}
	\label{fig:state_trajectory}
\end{figure}

\begin{figure}[!t]
	\centering
	\includegraphics[width = 8.5cm]{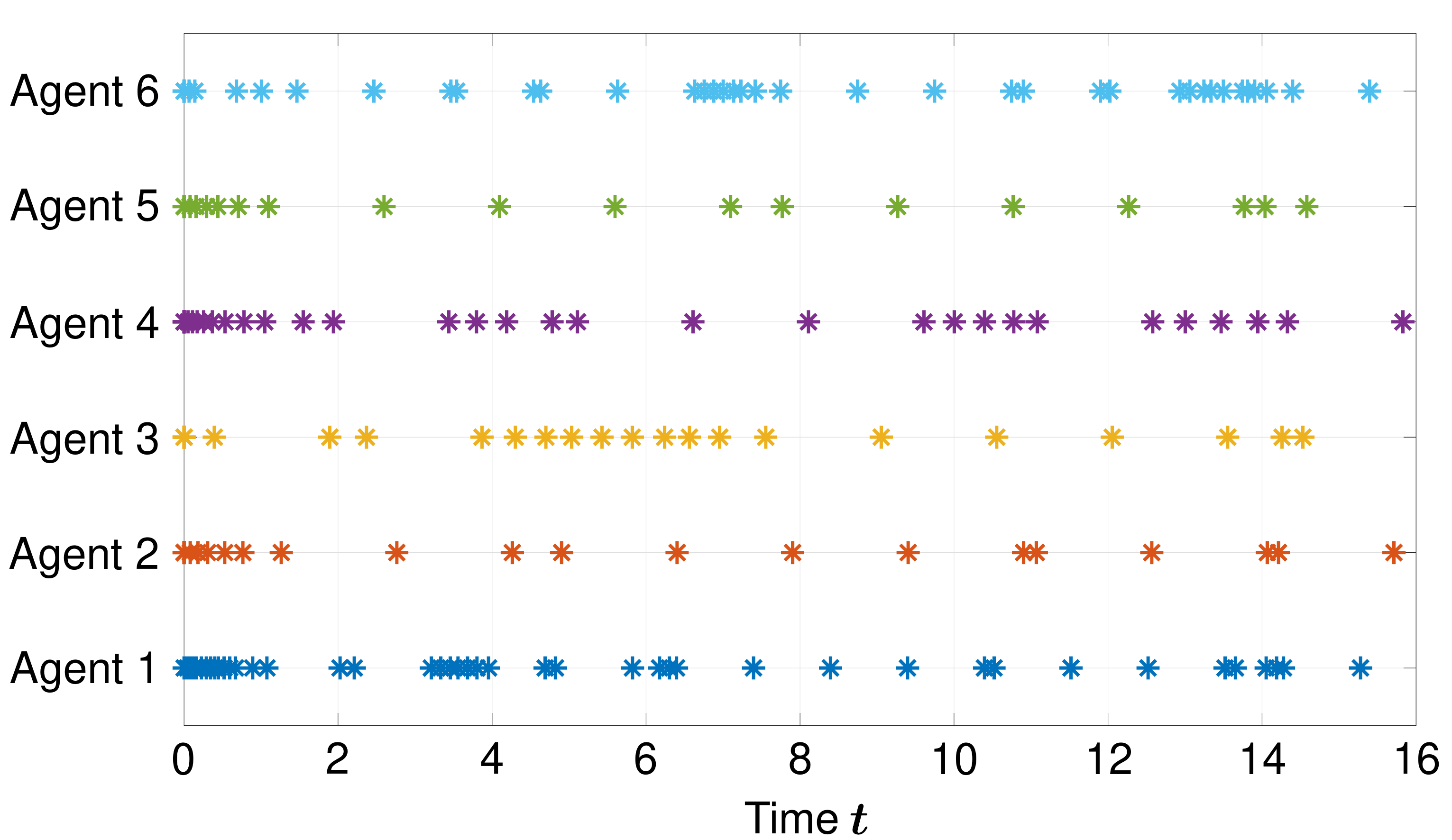}
	\caption{Sampling times.}
	\label{fig:sampling_time}
\end{figure}

\section{Conclusion}
\label{sec:conclusion}
We have proposed a joint design method of a finite-level dynamic quantizer
and a self-triggering mechanism
for asymptotic consensus
by relative state information.
The inter-event times are bounded from below by 
a strictly positive constant, and the sampling times
can be computed efficiently by using  the Lambert $W$-function.
The quantizer has been designed so that 
saturation is avoided and 
quantization errors exponentially converge to zero.
The new semi-norm introduced for the consensus analysis 
is constructed based on the maximum norm,
and the matrix exponential of the negative Laplacian matrix 
has the semi-contractivity property with respect to the semi-norm.
Future work will focus on 
extending the proposed method
to the case of directed graphs and agents with
high-order dynamics.

\section{Acknowledgments}
This work was supported by JSPS KAKENHI Grant Number
JP20K14362.

\appendix
\section*{Appendix A: Proof of Proposition~\ref{prop:Gamma_inf_bound}}
Let $0< \gamma \leq \lambda_2(L)$, and
let $\Lambda_0$, $\Lambda$, $V_0$, and $V$ 
be as in the proof of 
Lemma~\ref{lem:Ptau_bound}.

a) 
The inequality 
\[
2-\frac{2}{N} \leq \Gamma_{\infty}(\gamma)
\]
has already been proved in
\eqref{eq:Gamma_inf_bound}.
It remains to show that
\[
\Gamma_{\infty}(\gamma) \leq N-1.
\]
Since $V_0 \in \mathbb{R}^{N \times N}$ 
is orthogonal, 
we have $\|V_0\|_{\infty} \leq \sqrt{N}$.
Hence,
\[
\|V\|_{\infty} = \|V_0\|_{\infty} - \frac{1}{\sqrt{N}} \leq \sqrt{N} - \frac{1}{\sqrt{N}}.
\]
Moreover,
$
\|V^{\top}\|_{\infty}  \leq \|V_0^{\top}\|_{\infty}  = \sqrt{N}
$
and
\[
C \coloneqq \sup_{t \geq 0}\| e^{\gamma t}e^{-\Lambda t} \|_{\infty} \leq 1.
\]
Therefore, the inequality \eqref{eq:Gamma_bound} yields
\[
\Gamma_{\infty}(\gamma)
\leq C \|V\|_{\infty} \|V^{\top}\|_{\infty} 
\leq \left(
\sqrt{N} - \frac{1}{\sqrt{N}}
\right) \sqrt{N}  = N - 1.
\]

b)
Suppose that $G$ is a complete graph. Then
\[
\Lambda_0
= \diag (0,\,N,\,\cdots,\,N).
\]
If $0< \gamma \leq \lambda_2(L) = N$, then
\[
\|e^{\gamma t}(e^{-Lt} - 
\mathbf{1}\bar{\mathbf{1}})\|_{\infty} \leq 
\|e^{N t}(e^{-Lt} - 
\mathbf{1}\bar{\mathbf{1}})\|_{\infty}
\]
for all $t \geq 0$.
Hence, it suffices by a) to show that 
\begin{equation}
\label{eq:gam_n_bound}
\sup_{t \geq 0}\|e^{N t}(e^{-Lt} - 
\mathbf{1}\bar{\mathbf{1}})\|_{\infty} = 
2- \frac{2}{N}.
\tag{A1}
\end{equation}

Using 
$
L = V_0\Lambda_0V_0^{\top}
$
and
\begin{equation}
\label{eq:one_one}
\mathbf{1}\bar{\mathbf{1}} = 
V_0 \,\diag (1,\, 0,\, \cdots,\, 0) \, V_0^{\top},
\tag{A2}
\end{equation}
we obtain
\begin{align*}
e^{N t}(e^{-Lt} - 
\mathbf{1}\bar{\mathbf{1}}) 
&=
V_0\left(e^{Nt} e^{-\Lambda_0 t}-\diag (e^{Nt},\,0,\,\cdots,\, 0)\right)V_0^{\top} \\
&=V_0\, \diag (0,\,1,\,\cdots,\, 1) \, V_0^{\top}
\end{align*}
for all $t \geq 0$.
Moreover, \eqref{eq:one_one} yields
\begin{align*}
V_0 \, \diag (0,\,1,\,\cdots,\, 1)\, V_0^{\top} &=
V_0 V_0^{\top} - V_0\,  \diag (1,\, 0,\, \cdots,\, 0) \, V_0^{\top} \\
&= 
I - \mathbf{1}\bar{\mathbf{1}}.
\end{align*}
Thus, \eqref{eq:gam_n_bound} holds by $\|I - \mathbf{1}\bar{\mathbf{1}}\|_{\infty}
= 2-2/N$. \qed

\section*{Appendix B: Proof of Proposition~\ref{prop:W1_case}}
Define the function $H$ by
\[
H(\tau) \coloneqq \tau + be^{-\tau} - c,\quad \tau \in \mathbb{R}.
\]
Then 
\[
-\tau+c \geq b{e^{-\tau}} \quad \Leftrightarrow \quad 
H(\tau) \leq 0.
\]
Since
\[
H'(\tau) = 1 - be^{-\tau},
\]
it follows that $H'(\tau) = 0$ holds at $\tau = \log b$.
From the assumption $c < b$, we have $H(0) > 0$.
Therefore, there exists $\tau >0$ such that $H(\tau) \leq 0$ if and only if
\begin{equation}
\label{eq:bc_cond}
\log b>0 \quad \text{and} \quad H(\log b) \leq 0.
\tag{B1}
\end{equation}
Since
\[
H(\log b) = \log b + 1 - c,
\]
it follows that \eqref{eq:bc_cond} is equivalent to
\begin{equation}
\label{eq:bc_cond2}
1 < b \leq  e^{-1+c}.
\tag{B2}
\end{equation}
Hence,
\[
\inf \{
\tau >0:
-\tau + c \geq  b e^{-\omega \tau}\}  = \infty
\]
if \eqref{eq:bc_cond2} does not hold.

The inequality $-\tau+c \geq be^{-\tau}$ can be written as
\begin{equation}
\label{eq:bc_W_eq}
(\tau-c)e^{\tau-c} \leq -be^{- c}.
\tag{B3}
\end{equation}
Let 
$W_0$ and $W_{-1}$ be
the primary and secondary branch of the Lambert $W$-function, respectively. 
In other words, $W_0(y)$ and $W_{-1}(y)$ are
the solutions $x = x_0 \in [-1,0)$ 
and $x = x_{-1} \in (-\infty,-1]$ 
of the equation $xe^{x} = y$ for $y \in [-e^{-1},0)$, respectively.
For each $y \in [-e^{-1},0)$,
\begin{equation}
\label{eq:lambert_eqiv}
xe^{x} \leq y
\quad 
\Leftrightarrow
\quad 
W_{-1}(y) \leq x \leq W_{0}(y);
\tag{B4}
\end{equation}
see, e.g., \cite{Corless1996}.

Suppose that the condition \eqref{eq:bc_cond2} holds.
The expression \eqref{eq:bc_W_eq} 
and the equivalence \eqref{eq:lambert_eqiv}
show that 
$-\tau+c \geq be^{-\tau}$ if and only if
\[
W_{-1}(-be^{-c}) + c \leq \tau \leq W_{0}(-be^{- c}) + c.
\]
Note that $W_{-1}(-be^{-c})+c$ and $W_{0}(-be^{- c}) + c$
are the solutions of the equation $H(\tau) = 0$.
Since $H(0) > 0$,  both solutions are positive.
Thus,
\[
\inf \{
\tau >0:
-\tau + c \geq b e^{-\tau}\}  = W_{-1}(-be^{-c}) + c 
\]
is obtained.
\qed

\end{document}